\documentclass[12pt]{article}
\usepackage{amssymb}
\usepackage{amsmath}
\usepackage{amsfonts}
\usepackage{geometry,enumerate,bbm}
\usepackage{natbib}
\usepackage{epsfig}
\usepackage{float}
\usepackage[onehalfspacing]{setspace}
\usepackage{graphicx,color}
\usepackage{subcaption}
\usepackage{lscape}
\usepackage[flushleft]{threeparttable}
\usepackage{ulem}
\usepackage{scrextend}

\definecolor{dgreen}{rgb}{0,0.5,0}

\RequirePackage[colorlinks,citecolor=blue,urlcolor=blue,linkcolor=blue]{hyperref}
\allowdisplaybreaks

\newtheorem{theorem}{Theorem}

\newtheorem{lemma}{Lemma}

\newenvironment{proof}[1][Proof]{\noindent \textbf{#1.} }{\  \rule{0.5em}{0.5em}}
\renewcommand{\baselinestretch}{1.5}

\newcommand*\erasestuff[1]{}
\input{tcilatex}
\begin{document}

\title{Moment Conditions for Dynamic Panel Logit Models \\ with Fixed Effects\thanks{
We thank Francesca Molinari and three anonymous referees for constructive feedback
and suggestions.
We are also  grateful to St{\'e}phane Bonhomme, Runtong Ding, Geert Dhaene, Luojia Hu, Yoshitsugu Kitazawa, Shakeeb Khan,  Chris Muris, Whitney Newey, and numerous seminar
participants for useful
comments and discussions. Sharada
Dharmasankar provided excellent research assistance. This research was
supported by the Gregory C. Chow Econometric Research Program at Princeton
University, by the National Science Foundation (Award Numbers 1824131 and 2116630),
by the Economic and Social Research Council through the ESRC Centre for
Microdata Methods and Practice (grant numbers RES-589-28-0001, RES-589-28-0002 and ES/P008909/1), and by  the
European Research Council grants ERC-2014-CoG-646917-ROMIA and
ERC-2018-CoG-819086-PANEDA.}}
\author{\setcounter{footnote}{2}Bo E. Honor{\'e}\thanks{%
Princeton University and  The Dale T Mortensen
Centre at the University of Aarhus, \texttt{honore@princeton.edu} } \and Martin Weidner%
\thanks{%
University of Oxford, \texttt{martin.weidner@economics.ox.ac.uk} } }
\date{December 2023}

\maketitle
\thispagestyle{empty}
\setcounter{page}{0}

\renewcommand{\baselinestretch}{1.5}
\begin{abstract}
\noindent  
This paper investigates the construction of moment conditions in discrete choice panel data with individual specific fixed effects. We describe how to systematically explore the existence of moment conditions that do not depend on the fixed effects, and we demonstrate how to construct them when they exist. Our approach is closely related to the numerical “functional differencing” construction in \cite{bonhomme2012functional}, but our emphasis is to find explicit analytic expressions for the moment functions. We first explain the construction and give examples of such moment conditions in various models. Then, we focus on the dynamic binary choice logit model and explore the implications of the moment conditions for identification and estimation of the model parameters that are common to all individuals.
\end{abstract}
\renewcommand{\baselinestretch}{1.5}
\vskip 3cm

\section{Introduction}

This paper is concerned with estimation of the common parameters in
nonlinear panel data models with individual-specific fixed effects in
situations where the relevant asymptotics is an increasing number of
cross-sectional units observed over a fixed number of time periods.   Our contribution is a  general approach for constructing conditional moment conditions
when the dependent variable can take a finite number of values, and we
demonstrate how the approach can be used to construct moment conditions for
logit models with strictly exogenous explanatory variables as well as lagged
dependent variables.

The economic motivation for the econometric model investigated in this
paper is the question of whether persistence in economic data is due to
unobserved heterogeneity or state dependence. This question dates back to
\cite{Heckman78a}  and can be formulated as a distinction
between individual-specific fixed effects and lagged dependent variables.
This framework has proven relevant in many areas of economics. For example, 
\cite{PakesPorterShepardCalderWand2022}  have  
 employed it to
study the importance of switching costs in a model of health
insurance plan choice.

Econometrically, this paper makes a contribution to the literature on the
estimation of nonlinear econometric models with fixed effects. The challenge
is that if the fixed effects enter in a way that is not
additive or multiplicative, then one cannot simply difference or
quasi-difference it away as one would in a linear or
multiplicative model. At the same time, treating the fixed effects as
parameters to be estimated in a nonlinear model will generally lead to an inconsistent
estimator of the common parameters as the number of
cross-sectional units increases with the number of time periods fixed. 
This is what is known as the
incidental parameters problem. See \cite{neyman1948consistent}. One solution to the incidental parameters
problem in parametric models is to look for sufficient statistics for the
fixed effects. By definition, the conditional likelihood, conditional on
these sufficient statistics, will not depend on the fixed effects, so it can potentially be used for estimation.
This approach was pioneered by \cite{rasch1960studies} and \cite%
{andersen1970asymptotic}. Unfortunately, there are relatively few models for
which one can find such sufficient statistics, so an alternative approach is to
try to construct moment conditions that depend  on the parameters of
interest, but not on the individual-specific fixed effects. Papers by
\cite{Honore92, Honore1993}, \cite{Hu2002}, and \cite{johnson2004identification} are earlier specific examples of this, and \cite{bonhomme2012functional} developed a general approach for obtaining moment conditions
 via ``functional
differencing.''

Our paper operationalizes  the proposal in \cite{bonhomme2012functional}  for the case of dependent variables with a finite number of possible values.  
Specifically, we  offer  a  systematic method  for how to first
numerically explore the potential for constructing   moment conditions and then derive their analytic expressions. 
We then apply this machinery to create moment conditions for a prominent case: the fixed effects logit model
with strictly exogenous explanatory variables and lagged dependent
variables. For models with one lag, we  give  explicit expressions for all available moment conditions
when $T\ge 3$, where $T$ is the number of time periods in addition to those that give the initial conditions for the dependent variable.   We also  provide  all
the moment conditions for the case with two lagged dependent variables and $T=4$ and $5$, as well as with three lags and $T=5$. 
Notably, for the case of one lag and
three time periods (in addition to the one that delivers the initial), our conditions align with those previously found by
\cite{kitazawa2013exploration,kitazawa2016root}. 

Estimation of panel data binary response models dates back to \cite%
{rasch1960studies}, who noted that in a logit model with strictly
exogenous explanatory variables, one can make inference regarding the
remaining parameters by conditioning on the sums of the dependent variable
for each individual. \cite{Chamberlain1985} and \cite{magnac2000subsidised}
demonstrated that it is also possible to find sufficient statistics for the
individual-specific fixed effects in logit models where the only explanatory
variables are lagged outcomes, and the common parameters can then be
estimated by maximizing the conditional likelihood (conditional on the
sufficient statistic).  Unfortunately, the conditional likelihood approach
referenced above does not generally carry over to logit models that have
both lagged dependent variables and strictly exogenous explanatory
variables. However, as shown in \cite{honore2000panel}, this approach does
apply if one is also willing to condition on the vector of covariates being
equal across certain time periods. This leads to an estimator that is
asymptotically normal under suitable regularity conditions, but the rate of
convergence is slower than the usual $\sqrt{n}$ when there are continuous
covariates.
The logit assumption is crucial in the construction of the sufficient
statistics above. A number of papers (including \citealt{Manski87}, \citealt{aristodemou2018semiparametric} and \citealt{khan2019identification}) have
relaxed the logistic assumption. This literature suggests that point
estimation is sometimes possible without the logit assumption, and that
informative bounds can be constructed when it is not. On the other hand, 
\cite{chamberlain2010binary} showed that in a two-period static
threshold-crossing model  regular root-$n$ consistent estimation is only
possible in a logit setting. 
 This underpins the focus on the logistic assumption throughout this paper. 

\erasestuff{This paper relates to a number of strands of the literature. Estimation of
panel data binary response models dates back to \cite{rasch1960studies}, who
noted that in a logit model with strictly exogenous explanatory variables,
one can make inference regarding the remaining parameters by conditioning on
the sums of the dependent variable for each individual, which are the
sufficient statistics for the fixed effects. \cite{Manski87} showed that it
is possible to identify and consistently estimate a semiparametric version
of the model which relaxes the logistic assumption. \cite{Manski87}'s
estimator is not root-$n$ consistent, and \cite{chamberlain2010binary}
showed that regular root-$n$ consistent estimation is only possible in a
logit setting  when only two time periods are observed.
See also \cite{jochmans2017note}.}

\erasestuff{A number of papers have attempted to relax the assumption that the
explanatory variables are strictly exogenous by including lagged dependent
variables. Building on \cite{cox1958regression}, \cite{Chamberlain1985}  and \cite%
{magnac2000subsidised} demonstrated that it is possible to find sufficient
statistics for the individual-specific fixed effects in logit models where
the only explanatory variables are lagged outcomes. Conditioning on these
sufficient statistics leads to a likelihood function that does not depend on
the fixed effects, but they typically depend on some of the unknown
parameters of the model. The parameters can then be estimated by maximizing
the conditional likelihood. The resulting estimator is root-$n$ consistent
and asymptotically normal. The conditional likelihood approach referenced
above does not generally carry over to logit models that have both lagged
dependent variables and strictly exogenous explanatory variables. However,
as shown in \cite{honore2000panel}, this approach does apply if one is also
willing to condition on the vector of covariates being equal across certain
time periods.  This leads to an estimator that is asymptotically normal under
suitable regularity conditions, but the rate of convergence is slow when
there are continuous covariates.}

\erasestuff{Papers by \cite{honore2000panel}, \cite{aristodemou2018semiparametric} and
\cite{khan2019identification} relax the logistic assumption. This literature
suggests that point estimation is sometimes possible  and that informative
bounds can be constructed when it is not. On the other hand, the
impossibility result in \cite{chamberlain2010binary} suggests that the most
fruitful way to achieve regular root-$n$ consistent estimation is by
imposing a logistic assumption.\footnote{%
As an alternative to this, \cite{bartolucci2010dynamic} and \cite%
{Al-SadoonLiPesaran2017} have proposed models that are outside the framework
considered by \cite{chamberlain2010binary}.} This motivates the estimation
approach based on a logistic distribution assumption that we follow in the
current paper.}

 The paper is organized as follows: 
Section~\ref{SEC: Incidental parameter free moment conditions} details our approach for discerning and constructing moment conditions when the dependent variable assumes a finite number of values, also drawing on insights from \cite{dobronyi2021identification} regarding the number of available moments. 
In Section~\ref{sec:Derivation}, we showcase this within a panel data logit AR(1) model with \(T=3\) time periods. 
Section~\ref{sec:Examples} extends this by exploring various other models, including notably the dynamic ordered logit and dynamic multinomial models. This section also highlights the versatility of our method through examples like the panel data logit AR(1) with a heterogeneous time trend, and a static binary response model leveraging a mixture of logits. 
Section~\ref{sec:identitifactionAR1} discusses conditions under which the moment conditions  are guaranteed to identify the common parameters in the AR(1) logit model with strictly exogenous explanatory variables, 
while Section~\ref{SEC: Panel logit AR(1) model for general T>3}  demonstrates how to find moment conditions for the AR(1) logit model with strictly exogenous explanatory variables when the number of time periods is not three.  
Section~\ref{SEC: Empirical illustration} illustrates the usefulness of the approach by estimating a simple model for labor force participation, and Section~\ref{sec:conc} wraps up the discussion.

\section{Incidental parameter free moment conditions\label{SEC: Incidental parameter free moment conditions}}

\subsection{Model and moment conditions}

In this paper, we consider a panel data setting with $i=1,\ldots,n$ cross-sectional units
and $t=1,\ldots,T$ time periods. An econometrician models a sequence of
discrete outcomes, $Y_{i}=(Y_{i1},\ldots ,Y_{iT})$, as a function of
explanatory variables, $X_{i}=(X_{i1},\ldots ,X_{iT})$, initial conditions, $%
Y_{i}^{(0)}=(Y_{it}\,:\,t\leq 0)$,
and a time invariant \textquotedblleft fixed effect\textquotedblright , $A_{i}$,
as
\begin{equation}
\mathrm{Pr}\left( Y_{i}=y_{i}\,\Big|\,Y_{i}^{(0)}=y_{i}^{(0)},\,X_{i}=x_{i},%
\;A_{i}=\alpha _{i}\right) =f\big(y_{i}\,\big|\,y_{i}^{(0)},\,x_{i},\,\alpha
_{i};\,\theta \big).  \label{MainModelRestriction}
\end{equation}%
The function $f$ is assumed to be known up to the finite dimensional
parameter $\theta $.
The variables $Y_i$, $Y_{i}^{(0)}$
and $X_i$ are observed, but
 $A_i$
is unobserved.
The corresponding  conditional probabilities that can be identified
from the observed data are
\begin{equation*}
\mathrm{Pr}\left( Y_{i}=y_{i}\,\Big|\,Y_{i}^{(0)}=y_{i}^{(0)},\,X_{i}=x_{i}%
\right) =\int \,f\big(y_{i}\,\big|\,y_{i}^{(0)},\,x_{i},\,\alpha
_{i};\,\theta \big)\;g\big(\alpha _{i}\,\big|\,\,y_{i}^{(0)},\,x_{i}\big)%
\;d\alpha _{i},
\end{equation*}%
where the probability mass or density function, $g\big(\alpha _{i}\,\big|%
\,\,x_{i},\,y_{i}^{(0)}\big)$, of $A_{i}$ conditional on $X_{i}$ and $%
Y_{i}^{(0)}$ is left unspecified. 
We use $\mathcal{Y}$ to denote the set of
possible values of $Y_{i}=(Y_{i1},\ldots ,Y_{iT})$, which will be a finite
set in all the models considered in this paper.

Throughout this paper we assume that that $%
(Y_{i}^{(0)},Y_{i},X_{i},A_{i})$ are independent and identically distributed
across $i=1,\ldots ,n$, and our goal is to estimate the common parameters $%
\theta $ from the observed data as $n \rightarrow \infty$
and $T$ is fixed.
The difficulty in identifying and estimating $\theta $ is that the
individual specific fixed effects $(\alpha _{1},\ldots ,\alpha _{n})$, or
equivalently their unknown conditional distribution $g\big(\alpha _{i}\,\big|%
\,\,y_{i}^{(0)},\,x_{i}\big)$, constitute a high-dimensional nuisance
parameter, that is, we are faced with a classic \cite{neyman1948consistent}
 incidental parameter problem.
The leading example considered throughout most of this paper is the binary
choice logit AR(1) model, where $Y_{it}\in \{0,1\}$, $X_{it}\in \mathbb{R}%
^{K}$, $A_{i}\in \mathbb{R}$, and the model restriction reads
\begin{equation}
\mathrm{Pr}\left( Y_{it}=1\,\big|\,Y_{i}^{t-1},X_{i},A_{i}\right) =\frac{%
\exp \big(X_{it}^{\prime }\,\beta +Y_{i,t-1}\,\gamma +A_{i}\big)}{1+\exp %
\big(X_{it}^{\prime }\,\beta +Y_{i,t-1}\,\gamma +A_{i}\big)},  \label{model}
\end{equation}%
with $Y_{i}^{t-1}=(Y_{i,t-1},Y_{i,t-2},\ldots )$, $\beta \in \mathbb{R}^{K}$
and $\gamma \in \mathbb{R}$. In this example, we have $\theta =(\beta
,\gamma )$, $Y_{i}^{(0)}=Y_{i0}$, and
\begin{align}
f\big(y_{i}\,\big|\,y_{i}^{(0)},\,x_{i},\,\alpha _{i};\,\theta \big)&
=\prod_{t=1}^{T}\frac{\exp \left(y_{it} \left( x_{it}^{\prime
}\,\beta +y_{i,t-1}\,\gamma +\alpha _{i}\right) \right)}{1+\exp \left( x_{it}^{\prime
}\,\beta +y_{i,t-1}\,\gamma +\alpha _{i}\right)  }  \notag \\[10pt]
& =:p(y_{i}, y_{i}^{(0)},x_{i},\beta ,\gamma ,\alpha _{i}).
\label{DefProb}
\end{align}%
In this binary choice example, the set of possible outcomes $\mathcal{Y}%
=\{0,1\}^{T}$ has cardinality $|\mathcal{Y}|=2^{T}$.

A very general method to overcome the incidental parameter problem in the
models considered here is to find moment functions $%
m(y_{i},y_{i}^{(0)},x_{i},\theta )$ (different from  zero) such that the model restriction %
\eqref{MainModelRestriction} implies that for any  true 
parameter value $\theta$,
\begin{equation}
\mathbb{E}\left[ m(Y_{i},Y_{i}^{(0)},X_{i},\theta )\right] = 0.
\label{MomentsUnconditional}
\end{equation}%
If we can find such valid moment functions, then they can typically be used
to study identification of the parameter $\theta $ and to estimate it using
generalized method of moments (GMM). 
 The main challenge in this process is
to find such valid moment functions for a given panel model of interest.

 In the absence of any further restriction on the distribution of $%
(Y_{i}^{(0)},X_{i},A_{i})$, the unconditional moment restriction %
\eqref{MomentsUnconditional} can only be a consequence of the model %
\eqref{MainModelRestriction} if the conditional moment restriction
\begin{equation}
\mathbb{E}\left[ m(Y_{i},Y_{i}^{(0)},X_{i},\theta )\Big|%
\,Y_{i}^{(0)}=y_{i}^{(0)},\,X_{i}=x_{i},\;A_{i}=\alpha _{i}\right] =0
\label{MomentsConditional}
\end{equation}%
holds for all possible realizations $y_{i}^{(0)}$, $x_{i}$, $\alpha _{i}$.
Under weak regularity conditions, \eqref{MomentsUnconditional} then follows
from \eqref{MomentsConditional} by the law of iterated expectations.
Furthermore, \eqref{MomentsConditional} can be rewritten as
\begin{equation}
\sum_{y_{i}\in \mathcal{Y}}m(y_{i},y_{i}^{(0)},x_{i},\theta )\,f\big(y_{i}\,%
\big|\,y_{i}^{(0)},\,x_{i},\,\alpha _{i};\,\theta \big)=0,
\label{MomentsConditional2}
\end{equation}%
which shows that knowledge of $f\big(y_{i}\,\big|\,y_{i}^{(0)},\,x_{i},\,%
\alpha _{i};\,\theta \big)$ is sufficient to verify %
\eqref{MomentsConditional}, and therefore \eqref{MomentsUnconditional}, for
a given moment function, $m$. 

Consider a single moment function $m(y_{i},y_{i}^{(0)},x_{i},\theta )\in
\mathbb{R}$ and fixed values of $y_{i}^{(0)}$, $x_{i}$, $\theta $. Then, for
every value of $\alpha _{i}$, the condition \eqref{MomentsConditional2}
constitutes one linear restriction on the vector $%
[m(y_{i},y_{i}^{(0)},x_{i},\theta )\,:\,y_{i}\in \mathcal{Y}]\in \mathbb{R}%
^{|\mathcal{Y}|}$.  Finding $m(y_{i},y_{i}^{(0)},x_{i},\theta )\in \mathbb{R}$ for
fixed $y_{i}^{(0)}$, $x_{i}$, and $\theta $ then requires solving an infinite
number of linear equations in $|\mathcal{Y}|$ variables. Depending on the
choice of $f\big(y_{i}\,\big|\,y_{i}^{(0)},\,x_{i},\,\alpha _{i};\,\theta %
\big)$ no solution may exist to this infinite dimensional system of
equations. The key finding of this paper is that dynamic logit models do
generally have solutions to this system, that is, moment conditions of the
form \eqref{MomentsUnconditional} are generally available in such models.

\subsection{Strategy for exploring and using such moment conditions\label%
{Sec: Strategy for exploring and using such moment conditions}}

In this subsection, we briefly outline a three-step strategy for obtaining
valid moment conditions of the form \eqref{MomentsUnconditional} for a model
$f\big(y\,\big|\,y^{(0)},\,x,\,\alpha ;\,\theta \big)$ with $T$ time
periods. We drop all indices $i$ unless they are explicitly required.

The first step is to determine numerically whether it seems likely that one
can find moment functions that satisfy \eqref{MomentsConditional}. To do
this, we choose numerical values for $y^{(0)}$, $x$, and $\theta $, and also
choose $Q>|\mathcal{Y}|$ different numerical values for the fixed effects $%
(\alpha _{1},\ldots ,\alpha _{Q})\subset \mathcal{A}^{Q}$. We then check
numerically whether for those values, the system
\begin{equation}
\sum_{y\in \mathcal{Y}}m(y)\,f\big(y\,\big|\,y^{(0)},\,x,\,\alpha
_{q};\,\theta \big)=0,\qquad q=1,\ldots ,Q,  \label{MomentsConditional3}
\end{equation}%
of $Q$ equations in $|\mathcal{Y}|$ unknowns, $[m(y):y\in \mathcal{Y}]\in
\mathbb{R}^{|\mathcal{Y}|}$, has a solution other than $m=0$ (and if so, how
many). If the $Q$ equations have at least one solution, then one could repeat this
exercise for multiple randomly chosen numerical values of $y^{(0)}$, $x$, $%
\theta $ and of the fixed effects. In this step it is important to use
sufficient numerical precision in those calculations, see Appendix~\ref{sec:Computation} for
more details.

If the conclusion of the first step is that moment functions seem to exist,
then the next step is to find them. One way to proceed is by choosing
specific numerical values for $(\alpha _{1},\ldots ,\alpha _{Q})$, but now
solve the system \eqref{MomentsConditional3} analytically for arbitrary
values of $y^{(0)}$, $x$, $\theta $.\footnote{%
Since we consider discrete choice, the initial condition $y^{(0)}$ takes a
finite number of discrete values, and we can perform the analysis separately
for each of value of $y^{(0)}$. But for $x$ and $\theta $ we need to allow
for arbitrary general values in this step.} The corresponding solution for $%
m(y)$ will depend on $y^{(0)}$, $x$, and $\theta $, and we therefore denote the solution by $%
m(y,y^{(0)},x,\theta )$.
The solution will not depend on the specific numeric values of $\alpha
_{1},\ldots ,\alpha _{Q}$ if we have truly found a
valid moment condition for the chosen model.
See Section~\ref{sec:Derivation} for a
concrete example.

Since the moment functions, $m(y,y^{(0)},x,\theta )$, obtained in the second
step are obtained using a set of specific numerical values of $\alpha
_{1},\ldots ,\alpha _{Q}$, the third step is to verify analytically that
they satisfy the condition \eqref{MomentsConditional2} for all $\alpha \in \mathbb{R}$.

Once one has constructed moment functions, $m(y,y^{(0)},x,\theta )$, using
the strategy outlined  so far, the next step is to study the implications of
those moment functions for identification and estimation of $\theta $. It is
also useful to study the properties of the moment functions to obtain a
better understanding of their structure and origin. In particular, the first
two steps can only be implemented for a given number of time periods $T$.
However, by studying the moment functions obtained for specific choices of $%
T $, one may be able to draw general conclusions that make it possible to
write down all moment functions for a given model for all possible values of
$T$.

In the next section, we follow the steps outlined above to construct moment
conditions for the binary choice logit AR(1) model in equation (\ref{model}%
). However, there are many other interesting semi-parametric discrete choice
panel models $f\big(y\,\big|\,y^{(0)},\,x,\,\alpha ;\,\theta \big)$ for
which moment conditions of the form \eqref{MomentsUnconditional} exist, but
have not yet been studied systematically --- see Section~\ref{sec:Examples}
below for some concrete examples. The above work program can therefore be
seen as a blueprint for  an extensive research agenda beyond the current paper. We
have recently implemented this blueprint in \cite{honore2021dynamic} for the
case of dynamic ordered choice panel models. \cite{davezies2022fixed} can be
seen as another example that implements the above program for static binary
choice panel model with idiosyncratic error distributions that generalize
the logistic case in a particular way.

\cite{Dano2023arXiv} uses ``transition functions''  to derive moment conditions for
dynamic discrete choice logit panel models. This reproduces and extends various results in the current paper, with the advantage that the method more easily generalizes to an arbitrary number of time periods. \cite{Dano2023arXiv}  also works out the semiparametric efficiency bound for the AR(1) panel logit model with regressors.

\subsection{Lower bound on the number of moment conditions\label{SEC: Lower
bound on number of moment conditions}}

\cite{dobronyi2021identification} point out that it is sometimes possible to
determine a lower bound on the number of moment conditions that can be
derived for a given model. 
Specifically, for many  of the models considered
below we have $A \in \mathbb{R}$,
and one can write the probability distribution in (\ref{MainModelRestriction}) as
\begin{equation*}
f\big(y\,\big|\,y^{(0)},\,x,\,\alpha ;\,\theta \big)=\kappa \left( a \right) \,
\dsum\limits_{k=1}^{K}a^{k-1}c_{k}\left( y \right)
\end{equation*}%
for some $K \in \{1,2,\ldots\}$, some positive function $\kappa$ of $%
a=\exp \left( \alpha \right) $ that does not depend on $y$,
and some functions $c_k$ of $y$ that do not depend on $a$.
Here, the functions $\kappa$
and $c_k$ also depend on
$y^{(0)}$, $x$, $\theta$,
but analogous to our discussion in the last subsection, those arguments are dropped to focus more clearly on the
dependence on $\alpha$ and $y$.

A moment function must then satisfy
\begin{equation}
\sum_{y\in \mathcal{Y}}m(y)\,\dsum\limits_{k=1}^{K}a^{k-1}c_{k}\left(
y\right) =0,\qquad \text{for all }a\in \left( 0,\infty \right)
\label{EQ: Polynomial}
\end{equation}%
which is equivalent to
\begin{equation*}
\sum_{y\in \mathcal{Y}}\,m(y)\,c_{k}(y)=0,\qquad \text{for all $k\in
\{1,\ldots ,K\}$.}
\end{equation*}%
These are $K$ linear conditions in $\left\vert \mathcal{Y}\right\vert $
unknown parameters $m(y)$. We therefore have at least $\left\vert \mathcal{Y}%
\right\vert -K$ linearly independent solutions $m(y)$. In other words, the
model must have at least $\left\vert \mathcal{Y}\right\vert -K$ conditional
moment conditions (conditional on the initial conditions and on the
explanatory variables). Of course, there is no guarantee that all of these
moment conditions are functions of the common parameter, $\theta $.

\section{Moment functions for the $T=3$ logit AR(1) model}
\label{sec:Derivation}

In this section, we apply the strategy outlined in Section \ref{Sec:
Strategy for exploring and using such moment conditions} to construct moment
conditions for the binary choice logit AR(1) model in equation (\ref{model})
when $T$ is three. In most applications, this corresponds to a total of four
time periods: three for which the models is assumed to apply, plus one that
delivers the initial condition, $y_{0}$.

\subsection{Verifying existence of moment functions numerically}

By numerically evaluating whether solutions to \eqref{MomentsConditional3}
exist for this model, one finds that $T=3$ is the smallest number of time
periods for which non-zero valid moment functions are available. Our
discussion in Section \ref{Moments when T=2} below formally shows that it is indeed
not possible to derive moment conditions when $T=2$. This is the reason why
we focus on $T=3$ in this section. For $T=3$ and $\gamma \neq 0$, one
furthermore finds by numerical experimentation that for each value of the
initial condition $y_{0}$, there exist exactly two linearly independent
moment functions that satisfy \eqref{MomentsConditional3}.

\subsection{Finding analytical solutions for the moment functions}

\label{sec:FindSolutions}

Having verified the existence of moment functions numerically, the next goal
is to find analytic formulas for them. That is, we want to find functions $%
m(y,y_{0},x,\beta ,\gamma )$ that satisfy \eqref{MomentsConditional2}.

Since $T=3$, we have $|\mathcal{Y}|=2^{T}=8$. We define vectors in $\mathbb{R%
}^{8}$ for the model probabilities and for the candidate moment functions:
\begin{equation*}
\hspace{-3pt}
\mathbf{p}(y_{0},x,\beta ,\gamma ,\alpha )
\hspace{-3pt}
=
\hspace{-3pt}
\left(
\begin{array}{@{}c@{}}
p((0,0,0),y_{0},x,\beta ,\gamma ,\alpha ) \\
p((0,0,1),y_{0},x,\beta ,\gamma ,\alpha ) \\
p((0,1,0),y_{0},x,\beta ,\gamma ,\alpha ) \\
p((0,1,1),y_{0},x,\beta ,\gamma ,\alpha ) \\
p((1,0,0),y_{0},x,\beta ,\gamma ,\alpha ) \\
p((1,0,1),y_{0},x,\beta ,\gamma ,\alpha ) \\
p((1,1,0),y_{0},x,\beta ,\gamma ,\alpha ) \\
p((1,1,1),y_{0},x,\beta ,\gamma ,\alpha )%
\end{array}%
\right) \hspace{-3pt},\, \, \, \mathbf{m}(x,y_{0},\beta ,\gamma )
\hspace{-3pt}
=
\hspace{-3pt}
\left(
\begin{array}{@{}c@{}}
m((0,0,0),y_{0},x,\beta ,\gamma ) \\
m((0,0,1),y_{0},x,\beta ,\gamma ) \\
m((0,1,0),y_{0},x,\beta ,\gamma ) \\
m((0,1,1),y_{0},x,\beta ,\gamma ) \\
m((1,0,0),y_{0},x,\beta ,\gamma ) \\
m((1,0,1),y_{0},x,\beta ,\gamma ) \\
m((1,1,0),y_{0},x,\beta ,\gamma ) \\
m((1,1,1),y_{0},x,\beta ,\gamma )%
\end{array}%
\right) \hspace{-3pt} .
\end{equation*}

For simplicity, we drop the arguments $y_{0}$, $x$, $\beta $, and $\gamma $
for the rest of this subsection. They are all kept fixed in the following
derivation, and they are the same in the probability vector $\mathbf{p}%
(\alpha )=\mathbf{p}(x,y_{0},\beta ,\gamma ,\alpha )$ and in the moment
vector $\mathbf{m}=\mathbf{m}(x,y_{0},\beta ,\gamma )$. The probability
vector $\mathbf{p}(\alpha )$ as a function of $\alpha $ is given by the
model specification. A moment vector $\mathbf{m}\in \mathbb{R}^{8}$ with $%
\mathbf{m}\neq 0$ is valid if it satisfies $\mathbf{m}^{\prime }\,\mathbf{p}%
(\alpha )=0$ for all $\alpha \in \mathbb{R}$; that is, a valid moment vector
needs to be orthogonal to $\mathbf{p}(\alpha )$ for all values of $\alpha $.
If we can find such a valid moment vector, then its entries will provide
moment functions that satisfy equation \eqref{MomentsConditional}, because $%
\mathbf{m}^{\prime }\,\mathbf{p}(\alpha )$ is equal to $\mathbb{E}\left[
m(Y,Y_{0},X,\beta _{0},\gamma _{0})\,\big|\,Y_{0}=y_{0},\,X=x,\,A=\alpha %
\right] $.

Any valid moment vector also satisfies $\lim_{\alpha \rightarrow \pm \infty }%
\mathbf{m}^{\prime }\,\mathbf{p}(\alpha )=0$. Moreover, the model
probabilities $\mathbf{p}(\alpha )$ are continuous functions of $\alpha $
with $\lim_{\alpha \rightarrow -\infty }\mathbf{p}(\alpha )=\mathbf{e}%
_{1}=(1,0,0,0,\allowbreak 0,0,0,0)^{\prime }$ and $\lim_{\alpha \rightarrow +\infty }%
\mathbf{p}(\alpha )=\mathbf{e}_{8}=(0,0,0,0,0,0,0,1)^{\prime }$, where $%
\mathbf{e}_{k}$ denotes the $k$'th standard unit vector in eight dimensions.
From this, we conclude:

\begin{itemize}
\item[(1)] Any valid moment vector $\mathbf{m}$ satisfies $\mathbf{e}%
_1^{\prime }\, \mathbf{m} = 0$ and $\mathbf{e}_8^{\prime }\, \mathbf{m} = 0$.
\end{itemize}

Furthermore, from our ``step 1'' analysis with concrete numerical values we
already know that:\footnote{%
The numerical experiment does not provide a proof of this, but we still take
this as an input in our moment condition derivation, with the final
justification given by Lemma~\ref{lemma:moments_p1T3}.}

\begin{itemize}
\item[(2)] There are two linearly independent solutions to the equations $%
\mathbf{m}^{\prime }\,\mathbf{p}(\alpha )=0$ for all $\alpha \in \mathbb{R}$%
. This implies that the set of valid moment vectors $\mathbf{m}$ is
two-dimensional.
\end{itemize}

Motivated by hypothesis (2), the aim is to find two linearly independent
moment vectors $\mathbf{m}^{(0)}$ and $\mathbf{m}^{(1)}$ for each $y_{0}\in
\{0,1\}$. To distinguish $\mathbf{m}^{(0)}$ and $\mathbf{m}^{(1)}$ from each
other, we impose the condition $\mathbf{e}_{7}^{\prime }\,\mathbf{m}^{(0)}=0$
for the first vector and the condition $\mathbf{e}_{2}^{\prime }\,\mathbf{m}%
^{(1)}=0$ for the second vector. In addition, we require a normalization for
each of these vectors, because an element of the nullspace can be multiplied
by an arbitrary nonzero constant to obtain another element of the nullspace.
We choose the normalizations $\mathbf{e}_{4}^{\prime }\,\mathbf{m}^{(0)}=-1$
and $\mathbf{e}_{5}^{\prime }\,\mathbf{m}^{(1)}=-1$. Together with the
conditions in (1), this specifies four affine restrictions on each of the
vectors $\mathbf{m}^{(0)},\mathbf{m}^{(1)}\in \mathbb{R}^{8}$. To define $%
\mathbf{m}^{(0)}$ and $\mathbf{m}^{(1)}$ uniquely, we require four more
affine conditions for each. We therefore choose four numeric values $\alpha
_{q}$ and impose the orthogonality between $\mathbf{p}(\alpha _{q})$ and $%
\mathbf{m}^{(0/1)}$. Thus, motivated by (1) and (2), we need to solve the
following two linear systems of equations:

\begin{itemize}
\item[(0)] $\mathbf{e}_1^{\prime }\, \mathbf{m}^{(0)}=0$, \; \; $\mathbf{e}%
_8^{\prime }\, \mathbf{m}^{(0)}=0$, \; \; $\mathbf{e}_7^{\prime }\, \mathbf{m%
}^{(0)}=0$, \; \; $\mathbf{e}_4^{\prime }\, \mathbf{m}^{(0)}=-1$, \newline
$\mathbf{p}^{\prime }(\alpha_q) \, \mathbf{m}^{(0)}=0$, \;\; for $q=1,2,3,4$.

\item[(1)] $\mathbf{e}_1^{\prime }\, \mathbf{m}^{(1)}=0$, \; \; $\mathbf{e}%
_8^{\prime }\, \mathbf{m}^{(1)}=0$, \; \; $\mathbf{e}_2^{\prime }\, \mathbf{m%
}^{(1)}=0$, \; \; $\mathbf{e}_5^{\prime }\, \mathbf{m}^{(1)}=-1$, \newline
$\mathbf{p}^{\prime }(\alpha_q) \, \mathbf{m}^{(1)}=0$, \;\; for $q=1,2,3,4$.
\end{itemize}

If it is indeed possible to find such moment functions $\mathbf{m}^{(0/1)}$,
then it must be possible for the four values of $\alpha $ to be chosen
arbitrarily without affecting the solutions $\mathbf{m}^{(0/1)}$. For
example, $\alpha _{q}=q$ is a valid choice. Note that the
two-dimensional span of the vectors $\mathbf{m}^{(0)}$ and $\mathbf{m}^{(1)}$
and the potential of the moment conditions to identify and estimate $\beta $
and $\gamma $ are not affected by the normalizations $\mathbf{e}_{4}^{\prime
}\,\mathbf{m}^{(0)}=-1$, $\mathbf{e}_{5}^{\prime }\,\mathbf{m}^{(1)}=-1$, $%
\mathbf{e}_{7}^{\prime }\,\mathbf{m}^{(0)}=0$, and $\mathbf{e}_{2}^{\prime
}\,\mathbf{m}^{(1)}=0$.

The systems of linear equations (0) and (1) above uniquely determine $%
\mathbf{m}^{(0)}$ and $\mathbf{m}^{(1)}$. By defining the $8\times 8$
matrices $\mathbf{B}^{(0)}=[\mathbf{e}_{1},\mathbf{e}_{8},\mathbf{e}_{7},%
\mathbf{e}_{4},\allowbreak \mathbf{p}(\alpha _{1}),\mathbf{p}(\alpha _{2}),%
\mathbf{p}(\alpha _{3}),\mathbf{p}(\alpha _{4})]^{\prime }$, and $\mathbf{B}%
^{(1)}=[\mathbf{e}_{1},\mathbf{e}_{8},\mathbf{e}_{2},\mathbf{e}_{5},\mathbf{p%
}(\alpha _{1}),\mathbf{p}(\alpha _{2}),\mathbf{p}(\alpha _{3}),\mathbf{p}%
(\alpha _{4})]^{\prime }$, we can rewrite those systems of equations as $%
\mathbf{B}^{(0)}\,\mathbf{m}^{(0)}=-\mathbf{e}_{4}$, and $\mathbf{B}^{(1)}\,%
\mathbf{m}^{(1)}=-\mathbf{e}_{4}$. Solving this gives%
\begin{eqnarray}
\mathbf{m}^{(0)}\, &=&\,-\,\left( \mathbf{B}^{(0)}\right) ^{-1}\,\mathbf{e}%
_{4} \, , \nonumber \\
\mathbf{m}^{(1)}\, &=&\,-\,\left( \mathbf{B}^{(1)}\right) ^{-1}\,\mathbf{e}%
_{4} \, .
\label{FindMomentsAnalytical}
\end{eqnarray}%
Plugging the analytical expression for $\mathbf{p}(\alpha )=\mathbf{p}%
(x,y_{0},\beta ,\gamma ,\alpha )$ into the definitions $\mathbf{B}^{(1)}$
and $\mathbf{B}^{(0)}$, we thus obtain analytical expressions for $\mathbf{m}%
^{(0)}=\mathbf{m}^{(0)}(x,y_{0},\beta ,\gamma )$ and $\mathbf{m}^{(1)}=%
\mathbf{m}^{(1)}(x,y_{0},\beta ,\gamma )$.

To report the results, we denote the components of the solutions $\mathbf{m}%
^{(0/1)}=\mathbf{m}^{(0/1)}(x,y_0, \allowbreak \beta ,\gamma ) \in \mathbb{R}^8$ by $%
m^{(0/1)}(y,x,y_0,\beta ,\gamma ) \in \mathbb{R}$, for $y \in \mathcal{Y}$.
Furthermore, let $x_{ts}=x_{t}-x_{s}$. Then,   the solutions  are
\begin{align}
m^{(0)}(y,y_0,x,\beta ,\gamma )& =\left\{
\begin{array}{ll}
\exp \left( x_{23}^{\prime }\beta \right) -1 & \text{if }y=(0,0,1), \\
-1 & \text{if }(y_{1},y_{2})=(0,1), \\
\exp \left( x_{31}^{\prime }\beta - y_0 \, \gamma \right) & \text{if }%
y=(1,0,0), \\
\exp \left( x_{21}^{\prime }\beta + (1-y_0) \, \gamma \right) & \text{if }%
y=(1,0,1), \\
0 & \text{otherwise},%
\end{array}
\right.  \notag \\[5pt]
m^{(1)}(y,y_0,x,\beta ,\gamma )& =\left\{
\begin{array}{ll}
\exp \left( x_{12}^{\prime }\beta + y_0 \, \gamma \right) & \text{if }%
y=(0,1,0), \\
\exp \left( x_{13}^{\prime }\beta - (1-y_0) \,\gamma \right) & \text{if }%
y=(0,1,1), \\
-1 & \text{if }(y_{1},y_{2})=(1,0), \\
\exp \left( x_{32}^{\prime }\beta \right) -1 & \text{if }y=(1,1,0), \\
0 & \text{otherwise}.%
\end{array}%
\right.  \label{SolutionMomentsT3}
\end{align}
The two solutions in \eqref{SolutionMomentsT3} are closely related: If $%
Y_{t} $ is generated according to \eqref{model}, then $Z_{t}=1-Y_{t}$ is
also generated according to \eqref{model}, but with $X_{t}$ replaced by $%
-X_{t}$ and $A$ replaced by $A-\gamma $. The solutions $m^{(0)}$ and $%
m^{(1)} $ are symmetric in the sense that $m^{(0)}(y,y_0,x,\beta ,\gamma
)=m^{(1)}(1-y,1-y_0,-x,\beta ,\gamma )$.

\subsection{Verifying that the moment functions are valid}

The following lemma establishes that the moment functions, $%
m^{(0/1)}(y,y_0,x,\beta ,\gamma )$, displayed in \eqref{SolutionMomentsT3}
are indeed valid. For this, it is not relevant how the moment functions were
derived.

\begin{lemma}
\label{lemma:moments_p1T3} If the outcomes $Y=(Y_1,Y_2,Y_3)$ are generated
from model \eqref{model} with $T=3$ and true parameters $\beta_0$ and $%
\gamma_0$, then we have for all $q \in \{0,1\}$, $y_0 \in \{0,1\}$, $x \in
\mathbb{R}^{K \times 3}$, $\alpha \in \mathbb{R}$ that
\begin{align*}
\mathbb{E} \left[ m^{(q)}(Y,Y_0,X,\beta_0,\gamma_0) \, \big| \, Y_0=y_0, \,
X=x, \, A=\alpha \right] &= 0 .
\end{align*}
\end{lemma}

This lemma is a special case of Theorem~\ref{th:AR1moments} below.
However, one can prove this lemma more easily by direct calculation: just
plug-in the definition of the probabilities $p(y,y_{0},x,\beta _{0},\gamma
_{0},\alpha )$ and moments $m^{(0/1)}(y,y_{0},x,\beta _{0},\gamma _{0}$) to
show that
\begin{equation*}
\sum_{y\in \{0,1\}^{3}}\,p(y,y_{0},x,\beta _{0},\gamma _{0},\alpha
)\;m^{(0/1)}(y,y_{0},x,\beta _{0},\gamma _{0})=0.
\end{equation*}%
The details of this calculation are provided in Appendix~B.2.1 of \cite{Honore2022moment}.

\subsection{On the number of moment conditions\label{SEC: On the number of moment conditions}}

As explained in Section \ref{SEC: Lower bound on number of moment
conditions}, \cite{dobronyi2021identification} provide a method for deriving
 (a lower bound on)
the number of moment conditions for a given model. For the panel logit AR(1)
model, the probability distribution for $Y_{i}=(Y_{i1},\ldots ,Y_{iT})$
(conditional on $Y_{i0}$, $X_{i}$, $A_{i}$) is given by
\begin{equation*}
f\big(y\,\big|\,y^{(0)},\,x,\,\alpha ;\,\theta \big)\ =\prod_{t=1}^{T}\frac{%
\left[\exp \left( x_{t}^{\prime }\,\beta +y_{t-1}\,\gamma +\alpha \right)
\right]^{y_{it}}}{1+\exp \left( x_{t}^{\prime }\,\beta +y_{t-1}\,\gamma +\alpha
\right) }.
\end{equation*}%
With $a=\exp (\alpha )$ and $\pi _{t}(y_{t-1})=\exp [x_{i}^{\prime }\,\beta
+y_{t-1}\,\gamma ]$, we then have
\begin{align*}
f\big(y\,\big|\,y^{(0)},\,x,\,\alpha ;\,\theta \big)\ & =\prod_{t=1}^{T}%
\frac{\left[ a\,\pi _{t}(y_{t-1})\right] ^{y_{t}}}{1+a\,\pi _{t}(y_{t-1})}=%
\frac{\left[ a\,\pi _{1}(y_{0})\right] ^{y_{1}}}{1+a\,\pi _{1}(y_{0})}%
\,\prod_{t=2}^{T}\frac{\left[ a\,\pi _{t}(y_{t-1})\right] ^{y_{t}}}{1+a\,\pi
_{t}(y_{t-1})} \\
& =\frac{\left[ a\,\pi _{1}(y_{0})\right] ^{y_{1}}}{1+a\,\pi _{1}(y_{0})}%
\,\prod_{t=2}^{T}\frac{[1+a\,\pi _{t}(1-y_{t-1})]\,\left[ a\,\pi
_{t}(y_{t-1})\right] ^{y_{t}}}{[1+a\,\pi _{t}(0)][1+a\,\pi _{t}(1)]}=\kappa
(a)\,\cdot \widetilde{p}(y,a),
\end{align*}%
where we defined
\begin{align}
\kappa (a)& =\frac{1}{1+a\,\pi _{1}(y_{0})}\,\cdot \,\prod_{t=2}^{T}\,\frac{1%
}{[1+a\,\pi _{t}(0)][1+a\,\pi _{t}(1)]},  \notag \\
\widetilde{p}(y,a)& =\left[ a\,\pi _{1}(y_{0})\right] ^{y_{1}}%
\prod_{t=2}^{T}\left\{ [1+a\,\pi _{t}(1-y_{t-1})]\,\left[ a\,\pi
_{t}(y_{t-1})\right] ^{y_{t}}\right\} =\sum_{k=1}^{2T}\,a^{k-1}\,c_{k}(y).
\notag
\end{align}%
This has the exact structure of equation (\ref{EQ: Polynomial}) with $K=2T$
and $\left\vert \mathcal{Y}\right\vert =2^{T}$, so there must be at least $%
2^{T}-2T$ moment conditions.  When $T=3$,
 the lower bound on the number of
conditional moment conditions is $2^{T}-2T=2$, which is exactly
the number of moment conditions we found in Lemma~\ref{lemma:moments_p1T3} above.

\section{Examples of moment functions in other models}

\label{sec:Examples}

In this section, we briefly discuss some other fixed effects panel data
models with discrete outcomes, for which it is possible to use the approach
outlined in Section \ref{Sec: Strategy for exploring and using such moment
conditions} to derive moment conditions.  
 The goal of this section is to illustrate the broad applicability
of the moment condition approach, and it can be skipped
by a reader interested in the binary choice AR(1) panel model
only.

\subsection{Static binary choice models}

In a static panel binary response model with strictly exogenous regressors $%
X_{i}=(X_{i1}, \ldots   , \allowbreak X_{iT})$ and fixed effects $A_{i}$, the conditional
distribution of the outcomes $Y_{i}=(Y_{i1},\ldots ,Y_{iT})$ is given by
\begin{equation}
f\big(y_{i}\,\big|\,\,x_{i},\,\alpha _{i};\,\beta \big)=\prod_{t=1}^{T}\left[
F\left( x_{it}^{\prime }\,\beta +\alpha _{i}\right) \right] ^{y_{it}}\left[
1-F\left( x_{it}^{\prime }\,\beta +\alpha _{i}\right) \right] ^{1-y_{it}},
\label{StaticModel}
\end{equation}%
where $F(\cdot )$ is a cumulative distribution function. The distribution in %
\eqref{StaticModel} is a special case of \eqref{MainModelRestriction}. For
the logistic case, $F(\varepsilon )=[1+\exp (-\varepsilon )]^{-1}$, one can
use that $S_{i}=\sum_{t=1}^{T}Y_{it}$ is a sufficient statistic for $A_{i}$
to estimate $\beta $ via the conditional maximum likelihood estimator (CMLE)
that conditions on $S_{i}$, see \cite{rasch1960studies} and \cite%
{andersen1970asymptotic}. In fact, \cite{chamberlain2010binary} showed that
for $T=2$, and subject to weak regularity conditions, root-$n$-consistent
estimation of $\beta $ is only possible if $F(\varepsilon )$ is logistic.%
\footnote{%
This result is for $\varepsilon _{it}$ independent across $t$.
Generalizations that allow for dependence across $t$ are derived in \cite%
{magnac2004panel}.} This implies that non-trivial moment functions $%
m(y_{i},x_{i},\beta )$ are available for the $T=2$ static panel model if and
only if $F(\varepsilon )=[1+\exp (-\kappa \,\varepsilon +\mu )]^{-1}$, for
some constants $\kappa >0$ and $\mu \in \mathbb{R}$.

Interestingly, for the static panel model with $T=3$, one can allow for
distributions $F(\cdot )$ that are not logistic and still estimate the
parameter $\beta $ at $\sqrt{n}$ rate, that is, the $T=2$ result of \cite%
{chamberlain2010binary} does not apply in that case. In particular, for $T=3$%
, \cite{johnson2004identification} and \cite{davezies2022fixed} consider
distributions of the form $F(\varepsilon )=\left[ 1+w_{1}\exp (-\lambda
_{1}\,\varepsilon )+w_{2}\exp (-\lambda _{2}\,\varepsilon )\right] ^{-1}$,
with non-negative real-valued parameters $w_{1}$, $w_{2}$, $\lambda _{1}$, $%
\lambda _{2}$, and derive moment conditions for \eqref{StaticModel}. One can use
the machinery in Section \ref{SEC: Incidental parameter free moment
conditions} to show that the specification considered by \cite%
{johnson2004identification} and \cite{davezies2022fixed} is not the only
extension of the logistic distribution that provides moment conditions for
the $T=3$ static model. For example, consider the case where $F(\varepsilon )
$ is a mixture of two logistic distributions with the same variance:
\begin{equation}
F(\varepsilon )=\frac{\omega }{1+\exp (-\lambda \,\varepsilon +\mu _{1})}+%
\frac{1-\omega }{1+\exp (-\lambda \,\varepsilon +\mu _{2})},
   \label{LogitMixture}
\end{equation}%
where $\omega \in \lbrack 0,1]$ is a mixture weight, $\lambda >0$
parametrizes the common variance of the logistic components, and $\mu
_{1},\mu _{2}\in \mathbb{R}$ parametrize the mean of the two components.
When plugging \eqref{LogitMixture} into \eqref{StaticModel} and then
applying the procedure described in  
Section \ref{Sec: Strategy for exploring and using such moment conditions}
to derive valid
moment functions, one finds that for $\mu _{1}\neq \mu _{2}$ and $\omega \in
(0,1)$ exactly one moment function exists 
 for general values of $\beta$ and $x$
when $T=3$.
This moment condition
is given by
\begin{align}
m(y,x,\beta )& =\sum_{(t,s,r)\in \mathcal{P}}\mathbbm{1}\left\{
(y_{t},y_{s})=(0,1)\right\} \;(1-2\,y_{r})\;\mathrm{sgn}(t,s,r)\;\exp
[\lambda \,(x_{t}-x_{s})^{\prime }\beta ]  \notag \\
& \qquad \qquad   \times \Big\{\omega \,\exp [(1-y_{r})(\mu _{2}-\mu
_{1})]+(1-\omega )\,\exp [y_{r}\,(\mu _{2}-\mu _{1})]\Big\},
\label{MomentFunctionStaticLogitMixtureT3}
\end{align}%
where $\mathcal{P}$ is the set of all six permutations of $(1,2,3)$, and for
$(t,s,r)\in \mathcal{P}$ the signature of that permutation is denoted by $%
\mathrm{sgn}(t,s,r)$.

In addition,
we have found numerically that one can allow for more general finite
mixtures of logistic distributions when $T$ exceeds $3$. For example, it
appears that one can allow for a mixture of three logistics when $T$ is 4,
six when $T=5$, ten when $T=6$, and eighteen when $T$ is 7. Calculations
like the ones in Section \ref{SEC: Lower bound on number of moment
conditions}  suggest  that if the number of mixtures is $Q$, then there are $%
2^{T}-TQ-1$ non-trivial conditional moment conditions in this model. For
example, if $T$ is 9 then there will be seven moment conditions when $F$ is
a mixture of 56 logistic cumulative distribution functions. We leave it to
future research to derive these and to investigate the extent to which they
identify the common parameters of the model.

\subsection{Fixed Effect Logit AR($p$) Models With $p>1$}
\label{subsecARp}

The analysis of the dynamic panel data logit model in 
 Section~\ref{sec:Derivation} generalizes
to a model with more than one lag. Specifically, consider the model
\begin{equation}
\mathrm{Pr}\left( Y_{it}=1\,\big|\,Y_{i}^{t-1},X_{i},A_{i},\beta ,\gamma
\right) =\frac{\exp \big(X_{it}^{\prime }\,\beta +\sum_{\ell
=1}^{p}\,Y_{i,t-\ell }\,\gamma _{\ell }+A_{i}\big)}{1+\exp \big(%
X_{it}^{\prime }\,\beta +\sum_{\ell =1}^{p}\,Y_{i,t-\ell }\,\gamma _{\ell
}+A_{i}\big)},  \label{modelARp}
\end{equation}%
where $\gamma =(\gamma _{1},\ldots ,\gamma _{p})^{\prime }$. We assume that
the autoregressive order $p\in \{2,3,4,\ldots \}$ is known, and that
outcomes $Y_{it}$ are observed for time periods $t=t_{0},\ldots ,T$, with $%
t_{0}=1-p$. Thus, the total number of time periods for which outcomes are
observed is $T_{\mathrm{obs}}=T+p$, consisting of $T$ periods for which the
model applies and $p$ periods to observe the initial conditions. We maintain
the definition $Y_{i}=(Y_{i1},\ldots ,Y_{iT})$, but the initial conditions
are now described by the vector $Y_{i}^{(0)}=(Y_{i,t_{0}},\ldots ,Y_{i0})$.

Numeric calculations similar to those for the model with one lag suggest
that for a given value of $p$, one requires $T\geq 2+p$ (i.e.\ $T_{\mathrm{%
obs}}\geq 2+2p$) time periods to find conditional moment conditions that
hold without restrictions on the parameters or on the support of the
explanatory variables.\footnote{%
In addition to those general moment conditions, there are additional ones
that only become available for special values of the parameters and of the
regressors.} For
example, a model with $p=3$ lags requires a total of eight time periods;
three that provide the initial conditions for $Y_{it}$, and five for which
the model is assumed to apply. Numerical calculations also suggest that
the number of linearly independent moment conditions available for each
initial condition, $y^{(0)}$, is equal to\footnote{%
We have verified this for $p\in \{0,\ldots ,6\}$ and $T\in \{2+p,\ldots ,8\}$%
, but believe that this formula for the number of linearly independent
moments holds for all integers $p$, $T$ with $T\geq 2+p$. However, a general
proof of this conjecture is beyond the scope of this paper.} $%
2^{T}-(T+1-p)\,2^{p}$. In Appendix~\ref{sec:ARp}, we provide analytic formulas for all
the moment functions that can be obtained with $T\leq 5$. Specifically, when
$p=2$, we provide four moment functions for $T=4$ and sixteen for $T=5$. For
$p=3$, there are eight moment functions, while there are no general moment
functions when $p\geq 4$ and $T\leq 5$.
 Identification of the parameters $\beta$ and $\gamma$
for $p \leq 3$ from the moment conditions is discussed in Appendix Section~B.3 of \cite{Honore2022moment}.

The special case of an AR(2)\ logit model with fixed effects and no
explanatory variables was considered in \cite{honore2019identification}.
Numerical calculations in that paper suggested that the common parameters in
such a model are point identified for $T=3$ (i.e.\ $T_{\mathrm{obs}}=5$),
but no proof of identification was provided. Evaluating the moment functions
in  Appendix~\ref{sec:ARp} at $\beta =0$ makes it clear why $\gamma =(\gamma _{1},\gamma _{2})$
is identified and how one would estimate it. Specifically, if the outcomes $%
Y=(Y_{1},Y_{2},Y_{3})$ are generated from the AR(2) panel logit model
without explanatory variables, then we have, for all $y^{(0)}\in \{0,1\}^{2}$
and $\alpha \in \mathbb{R}$, that
\begin{equation*}
\mathbb{E}\left[ m_{y^{(0)}}(Y,\gamma _{0})\,\big|\,Y^{(0)}=y^{(0)},\,A=%
\alpha \right] =0,
\end{equation*}%
with moment functions given by
\begin{align*}
m_{(y_0,y_0)}(y,\!\gamma )&\!=\!\!\left\{
\begin{array}{@{}l@{\,}l}
1 & \text{if }y=(y_0,\!1\!\!-\!\!y_0,\!y_0), \\
e^{-\gamma _{1}} & \text{if }y=(y_0,\!1\!\!-\!\!y_0,\!1\!\!-\!\!y_0), \\
-1 & \text{if }(y_{1},\!y_{2})=(1\!\!-\!\!y_0,\!y_0), \\
0 & \text{otherwise},%
\end{array}%
\right. \!\! & m_{(1\!-y_0,y_0)}(y,\!\gamma )&\!=\!\!\left\{
\begin{array}{@{}l@{\,}l}
-1 & \text{if }(y_{1},\!y_{2})=(1\!\!-\!\!y_0,\!y_0), \\
e^{\gamma _{2}-\gamma _{1}} & \text{if }y=(y_0,\!1\!\!-\!\!y_0,\!1\!\!-\!\!y_0), \\
e^{\gamma _{2}} & \text{if }y=(y_0,\!1-y_0,\!y_0), \\
0 & \text{otherwise},%
\end{array}%
\right. 
\end{align*}
where $y_{0}\in \{0,1\}$.

The moment functions $m_{(0,0)}$ and $m_{(1,1)}$ are strictly monotone in $%
\gamma _{1}$ and do not depend on $\gamma _{2}$. Each of them therefore
identify the parameter $\gamma _{1}$. For a given value  
of $\gamma _{1}$, the moment function $m_{(0,1)}$ and $m_{(1,0)}$ are
strictly monotone in $\gamma _{2}$, and they therefore each identify the
parameter $\gamma _{2}$ once $\gamma _{1}$ has been identified. A GMM
estimator based on these moment will be root-n consistent under standard
regularity conditions.

\subsection{Panel logit AR(1) with heterogeneous time trends\label{Panel logit AR(1) with heterogeneous time trends}}

  For arbitrary regressors, no valid moment functions seem
to exist
when some of the elements of $\beta $ are replaced by fixed effects $\beta_i$. However, if one of
the explanatory variables is a linear time trend, then it is possible to
allow for the coefficient on this variable to differ arbitrarily across
observations. Specifically, consider the generalization of the model in
equation (\ref{model}) to%
\begin{equation*}
\mathrm{Pr}\left( Y_{it}=1\,\big|\,Y_{i}^{t-1},X_{i},A_{i}\right) =\frac{%
\exp \big(X_{it}^{\prime }\,\beta +Y_{i,t-1}\,\gamma +tD_{i}+A_{i}\big)}{%
1+\exp \big(X_{it}^{\prime }\,\beta +Y_{i,t-1}\,\gamma +tD_{i}+A_{i}\big)}.
\end{equation*}
By mimicking the calculations in Section \ref{SEC: On the number of moment conditions}, one finds  that at least $\ell_{\min} = 2^T - \frac{T}{3}\left( 2T^{2}-3T+7\right) $ moment conditions need to exist in this model. For $T\geq 9$ we have $\ell_{\min}>0$, that is,
  it must be the case that moment conditions for this model
exist.
See Appendix Section \ref{Additional Calculations for Time Trends}
for details. 
However,
such calculations only yield a lower bound on the number of  moment conditions.\footnote{
Numerically, we  find 126 linearly independent moment conditions for the model with additional strictly exogenous explanatory variables when $%
T=9$ . 
This is larger than
the lower bound of $\ell_{\min} = 86$  moment conditions
derived in Appendix Section~\ref{Additional Calculations for Time Trends}.
This illustrates that exploring the polynomial structure of the model probabilities
(as described in Section \ref{SEC: On the number of moment conditions}) 
does not always give the exact number of available moment conditions in binary logit models.
}
Numerically, we do not find any
moment conditions for $T \leq 8$ for general parameter values with $x_{it}
\neq 0$, but we do find two valid moment conditions for $T=8$ if $x_{it}=0$
(so there are no additional regressors in the model,   and $\gamma$ is the only common parameter). Both of these moment conditions depend on the parameter $\gamma$. In Appendix Section \ref{Additional Calculations for Time Trends}, we discuss these moment conditions.

\subsection{Extensions to dynamic ordered logit model and dynamic
multinomial logit models}

The methods described in this paper can also be applied to dynamic panel data versions of other ``textbook'' logit models.  In particular, \cite%
{honore2021dynamic} 
use the procedure outlined in Section \ref{Sec: Strategy
for exploring and using such moment conditions} to find moment
conditions for dynamic panel data {\it ordered} logit models,
and \cite{Dano2023arXiv} shows how to obtain moment conditions for dynamic panel data
 {\it multinomial}  logit models.

\section{Identification}

\label{sec:identitifactionAR1}

This section shows that the  moment conditions for the panel logit AR(1)
model in Lemma~\ref{lemma:moments_p1T3}
can be used to uniquely identify the parameters $\beta $ and $\gamma $
under appropriate support conditions on the regressor $X$. The following
technical lemma turns out to be very useful in showing this.

\begin{lemma}
\label{lemma:INVERSION} Let $K \in \mathbb{N}_0$. For every $s =
(s_1,\ldots,s_K) \in \{-,+\}^K$ let $g_s : \mathbb{R}^{K} \times \mathbb{R}
\rightarrow \mathbb{R}$ be a continuous function such that for all $%
(\beta,\gamma) \in \mathbb{R}^{K} \times \mathbb{R} $ we have

\begin{itemize}
\item[(i)] $g_s(\beta,\gamma)$ is strictly increasing in $\gamma$.

\item[(ii)] For all $k \in \{1,\ldots,K\}$: If $s_k = +$, then $%
g_s(\beta,\gamma)$ is strictly increasing in $\beta_k$.

\item[(iii)] For all $k \in \{1,\ldots,K\}$: If $s_k = -$, then $%
g_s(\beta,\gamma)$ is strictly decreasing in $\beta_k$.
\end{itemize}
Then, the system of $2^K$ equations in $K+1$ variables
\begin{align}
g_s(\beta,\gamma)=0 , \qquad \text{for all} \; s \in \{-,+\}^K ,
\label{systemEQ}
\end{align}
has at most one solution.
\end{lemma}

To explain the lemma, consider the case $K=1$,\footnote{%
Note that $K=0$ is trivially allowed in Lemma~\ref{lemma:INVERSION}. We
then have $s=\emptyset $ and $g_{\emptyset }:\mathbb{R}\rightarrow \mathbb{R}
$ is a single increasing function, implying that $g_{\emptyset }(\gamma )=0$
can at most have one solution.} when we have two scalar parameters $\beta
,\gamma \in \mathbb{R}$. The lemma then requires that the two functions $%
g_{+}(\beta ,\gamma )$ and $g_{-}(\beta ,\gamma )$ are both strictly
increasing in $\gamma $, and $g_{+}$ is also strictly increasing in $\beta $%
, while $g_{-}$ is strictly decreasing in $\beta $.
Thus, $g_{+}(\beta ,\gamma )=0$ gives a solution for  $\gamma=\gamma(\beta)$ that is strictly decreasing in $\beta$, while 
$g_{-}(\beta ,\gamma )=0$ gives a solution $\gamma(\beta)$ that is strictly increasing, implying that the joint solution must be unique.

Combining the
moment conditions
in Lemma~\ref{lemma:moments_p1T3}
with the result of
Lemma \ref{lemma:INVERSION} allows us to provide sufficient conditions for
point identification in panel logit AR(1) models with $T=3$. For that purpose   we define
the sets
\begin{align*}
\mathcal{X}_{k,+}& =\{x\in \mathbb{R}^{K\times 3}\,:\,x_{k,1}\leq
x_{k,3}<x_{k,2}\;\;\text{or}\;\;x_{k,1}<x_{k,3}\leq x_{k,2}\}, \\
\mathcal{X}_{k,-}& =\{x\in \mathbb{R}^{K\times 3}\,:\,x_{k,1}\geq
x_{k,3}>x_{k,2}\;\;\text{or}\;\;x_{k,1}>x_{k,3}\geq x_{k,2}\},
\end{align*}%
for $k\in \{1,\ldots ,K\}$.
The set $\mathcal{X}_{k,+}$ is the
set of possible regressor values $x\in \mathbb{R}^{K\times 3}$ such that
either $x_{k,1}\leq x_{k,3}<x_{k,2}$ or $x_{k,1}<x_{k,3}\leq x_{k,2}$; that
is, the $k$'th regressor takes its smallest value in time period $t=1$ and
its largest value in time period $t=2$.  Conversely, the set $%
\mathcal{X}_{k,-}$ is the set of possible regressor values $x\in \mathbb{R}%
^{K\times 3}$ for which the $k$'th regressor takes its largest value in time
period $t=1$ and its smallest value in time period $t=2$.

The motivation behind the definition of those sets
is that for $x \in \mathcal{X}_{k,\pm}$ our moment functions $m^{(0/1)}(y,y_0,x,\beta ,\gamma )$ defined in Section~\ref{sec:FindSolutions}
have convenient monotonicity properties in the parameters $\beta_k$.
For example, for $m^{(0)}$ we have
\begin{align}
      \frac{\partial \, \mathbb{E}\left[ m^{(0)}(Y,Y_0,X,\beta ,\gamma )\,\Big|\,Y_{0}=y_0,\; X=x \right] }
    {\partial \beta_k}
    &>0, \;\; \text{for $x \in  \mathcal{X}_{k,+}$} \; ,
    \nonumber \\
    &<0 , \;\; \text{for $x \in  \mathcal{X}_{k,-}$} \;.
    \label{MonotonicityEmoments}
\end{align}
This is because the parameter $\beta$ appears in $m^{(0)}(y,y_0,x,\beta ,\gamma )$ only through
$\exp \left( x_{23}^{\prime }\beta \right)$, $\exp \left( x_{31}^{\prime }\beta \right)$ and $\exp \left( x_{21}^{\prime }\beta \right) $;
  $x \in  \mathcal{X}_{k,+}$ (or $x \in  \mathcal{X}_{k,-}$) guarantees that
the differences $x_{k,2}-x_{k,3}$ and  $x_{k,3}-x_{k,1}$ and $x_{k,2}-x_{k,1}$
are all $\geq 0$ ($\leq 0$), with some of them   strictly positive (negative).
The moment function
$m^{(1)}(y,y_0,x,\beta ,\gamma )$ 
has exactly the opposite monotonicity properties in $\beta$.

Next, for any vector 
$s = (s_1,\ldots,s_K) \in \{-,+\}^{K}$ we define the set $\mathcal{X}_{s}=\bigcap_{k\in
\{1,\ldots ,K\}}\mathcal{X}_{k,s_{k}}$  and  the corresponding expected moment functions, for $q,y_0 \in \{0,1\}$,
\begin{align*}
\overline{m}_{y_{0},s}^{(q)}(\beta ,\gamma )& =\mathbb{E}\left[
m^{(q)}(Y,Y_0,X,\beta ,\gamma )\,\Big|\,Y_{0}=y_{0},\;X\in \mathcal{X}_{s}\right] .
\end{align*}
Because $\mathcal{X}_{s}$ is the intersection of the sets $\mathcal{X}_{k,\pm}$, the
expected moment functions $\overline{m}_{y_{0},s}^{(0/1)}(\beta ,\gamma )$ 
have monotonicity properties with respect to all the elements of $\beta$ specified by
the sign vector $s = (s_1,\ldots,s_K)$.
For example, \eqref{MonotonicityEmoments} implies that $\overline{m}_{y_0,s}^{(0)}(\beta ,\gamma )$
is strictly increasing in  $\beta_k$ if $s_k=+$, and strictly decreasing in $\beta_k$ if $s_k=-$, for all $k \in \{1,\ldots,K\}$.

\begin{theorem}
\label{th:id1} Let $q,y_{0}\in \{0,1\}$. Let the
outcomes $Y=(Y_{1},Y_{2},Y_{3})$ be generated from model \eqref{model} with $%
T=3$ and true parameters $\beta _{0}$ and $\gamma _{0}$. Furthermore, for
all $s\in \{-,+\}^{K}$ assume that
\begin{equation*}
\mathrm{Pr}\left( Y_{0}=y_{0},\;X\in \mathcal{X}_{s}\right) >0,
\end{equation*}
and that the expected moment function $\overline{m}_{y_{0},s}^{(q)}(\beta ,\gamma )$ is well-defined.\footnote{
We could always guarantee $\overline{m}_{y_{0},s}^{(q )}(\beta ,\gamma )$ to be well-defined
by modifying the definition of the set $\mathcal{X}_{s}$ to only contain bounded regressor values.
}
Then, the solution to
\begin{equation}
\overline{m}_{y_{0},s}^{(q )}(\beta ,\gamma )=0 \qquad \text{for all}%
\;s\in \{-,+\}^{K} 
  \label{SystemTheoremID}
\end{equation}%
is unique and given by $(\beta _{0},\gamma _{0})$. Thus, the parameters $\beta_0$ and $\gamma_0$ are point-identified
\end{theorem}

The proof of the theorem is provided in the appendix.
Note that only one of the moment functions $m^{(0)}$ or $m^{(1)}$ is required to derive identification
 in Theorem~\ref{th:id1},
and  only one of the initial conditions $y_0 \in \{0,1\}$ needs to be observed.
The key assumption in Theorem~\ref{th:id1}
is that we have enough variation in the observed regressor values $X=(X_1,X_2,X_3)$
to satisfy the condition $\mathrm{Pr}\left( Y_{0}=y_{0},\;X\in \mathcal{X}_{s}\right) >0$, for all $s \in \{+,-\}^K$.

Theorem~\ref{th:id1} achieves identification of $\beta$ and $\gamma$ via conditioning on the sets of regressor values $\mathcal{X}_{s}$, which all have positive Lebesgue measure. By contrast, the conditional likelihood approach in
\cite{honore2000panel} conditions on the set  $x_2=x_3$, which has zero  Lebesgue measure,
and therefore also often zero probability measure, implying that the resulting estimates for $\beta$ and $\gamma$ usually converge
at a rate slower than root-$n$.
In our approach here,
using the sample analogs of the moment conditions $\overline{m}_{y_{0},s}^{(0/1)}(\beta ,\gamma ) = 0$  for $s \in \{+,-\}^K$,
we   immediately obtain GMM estimates for $\beta$ and $\gamma$
that are   root-n consistent under standard  regularity conditions.

However, in practice, we do not actually recommend estimation via the moment conditions in Theorem~\ref{th:id1}, because by conditioning
on $X\in \mathcal{X}_{s}$ these moment conditions still only use a small subset of the available information in the data.
Instead, many more unconditionally valid moment conditions for $\beta$ and $\gamma$ can be obtained from  Lemma~\ref{lemma:moments_p1T3} (or from Theorem~\ref{th:AR1moments} 
below for $T>3$), resulting in potentially much more
efficient estimators for $\beta$ and $\gamma$, and Section~\ref{sec:Emp} describes how we implement such estimators in practice.
Nevertheless, from a theoretical perspective, the identification result in
Theorem~\ref{th:id1} is important, because it comprises a significant improvement over
existing results for dynamic panel logit models with explanatory variables.

\section{Panel logit AR(1) model for general $T \geq 3$}
\label{SEC: Panel logit AR(1) model for general T>3}

In Section~\ref{sec:FindSolutions} above we already found analytic
formulas for valid moment functions that are free
of the fixed effects for the panel logit AR(1) model with three time periods. In this section we discuss various generalizations
of this result, most importantly to $T>3$ time periods.
 Before presenting those  positive results,
we first briefly discuss a negative result for $T=2$
time periods.

\subsection{Impossibility of moment conditions when $T=2\label{Moments when T=2}$}

Here, we argue that it is not possible to derive moment conditions for model (\ref{model}) on the basis of two time
periods plus the initial condition, $y_{0}$.
If one could  construct such moment conditions for $T=2$ that
hold conditional on the individual specific effects $A$, then the corresponding moment functions
$m(y,y_0,x,\beta,\gamma) $ would satisfy
\begin{align}
\sum_{y \in \{ 0,1\}^2} \, p(y,y_0,x,\beta,\gamma,\alpha) \;
m(y,y_0,x,\beta,\gamma) = 0,
   \label{MomentT2}
\end{align}
for all $\alpha \in \mathbb{R}$.
In the limit $\alpha \rightarrow \infty$ the model probabilities become zero, except
for $p((1,1),y_0,x,\beta,\gamma,\alpha) \rightarrow 1$.  This implies
 $m((1,1),y_0,x,\beta,\gamma) =0$. Analogously, in the limit $\alpha \rightarrow -\infty$
we have $p((0,0),y_0,x,\beta,\gamma,\alpha) \rightarrow 1$, which implies $m((0,0),y_0,x,\beta,\gamma) =0$.
Thus, only $m((0,1),y_0,x,\beta,\gamma)$ and $m((1,0),y_0,x,\beta,\gamma)$ can be non-zero,
and \eqref{MomentT2} therefore implies that
\begin{align}
     \frac{m((0,1),y_0,x,\beta,\gamma)} {m((1,0),y_0,x,\beta,\gamma)}
     &=  - \frac{p((1,0),y_0,x,\beta,\gamma,\alpha)} {p((0,1),y_0,x,\beta,\gamma,\alpha)}
 \nonumber    \\
     &=- \exp
\left( \left( x_{1}-x_{2}\right) ^{\prime }\beta +\gamma y_{0}\right)
 \frac{1+\exp \left( x_{2}^{\prime }\beta +\alpha \right) }
   {1+\exp \left( x_{2}^{\prime }\beta +\gamma +\alpha \right) } .
 \label{MM T=2}
\end{align}
Unless $\gamma =0$, the right hand side of \eqref{MM T=2}
   will  always have a non-trivial dependence on $\alpha$,
implying that no  moment conditions can be constructed for $T=2$ (that are valid conditional on arbitrary $%
A=\alpha $). For $\gamma =0$ equation (\ref{MM T=2}) yields the moment
conditions implied by \cite{Rasch60}'s conditional likelihood.

The fact that there are no moment conditions when $T=2$ is consistent with the non-identification result in \cite{Chamberlain2023SERIES}.

\subsection{Master lemma for obtaining all moment conditions}
\label{subsec:MasterLemma}

One can work out analytic moment functions for
model (\ref{model})  with $T=4$
and $T=5$ using the same derivation method described in
Section~\ref{sec:FindSolutions} for  $T=3$.
These are relatively brute force calculations that only
require limited human input and creativity. 
However, once those analytic moment functions are
obtained, one can move on to study their common structure,
which leads to the following lemma that allows us to derive  all  the valid
moment conditions for the panel logit AR(1) model
for an arbitrary number of time periods. 

Before presenting the lemma, we introduce some
additional notation:
The cumulative distribution function 
of the logistic distribution is given by
$\Lambda(\xi):=[1+\exp(-\xi)]^{-1}$. 
In addition, we define  
the cyclical decrement function $\delta : \{1,2,3\} \rightarrow \{1,2,3\}$ by
$$
     \delta(t) := \left\{ \begin{array}{ll}   
       3 & \text{for } t=1, \\
       1 & \text{for } t=2, \\
       2 & \text{for } t=3.
     \end{array} \right. 
$$

\begin{lemma}
   \label{lemma:Markov}
    Let $\widetilde Y_1, \widetilde  Y_2, \widetilde  Y_3 \in \{0,1\}$ 
    be binary random variables, and $W_1, W_2, W_3$
    be random variables (or vectors)  such that
$
    W_1 
    \rightarrow  \widetilde Y_1
        \rightarrow  W_2
        \rightarrow  \widetilde  Y_2   
       \rightarrow   W_3
      \rightarrow \widetilde  Y_3 
$
is a Markov chain,
conditional on the random vector $(X,A)$.\footnote{
The Markov chain assumption means that the 
density of 
$(W_1,\widetilde Y_1, W_2, \widetilde  Y_2, W_3, \widetilde  Y_3)$, conditional on $(X,A)$,
can be written as a product  
$ f_{\widetilde  Y_3|W_3,X,A}  \, 
  f_{W_3|\widetilde Y_2,X,A}   \, 
  f_{\widetilde  Y_2|W_2,X,A}  \, 
  f_{W_2|\widetilde Y_1,X,A} 
   f_{\widetilde  Y_1|W_1,X,A}  \, 
  f_{W_1|X,A} $.
}
Assume furthermore that
$p_t(\widetilde y_t \, |\, w_t,x,\alpha):={\rm Pr}\big( \widetilde Y_t = \widetilde y_t \, \big| \, W_t=w_t, 
\, X=x,
\, A=\alpha  \big)$
satisfies
$0<p_t(\widetilde y_t \, |\,w_t, x,\alpha)<1$,
for all $\widetilde y_t$, $w_t$, $x$, $\alpha$,
and
$t \in \{1,2,3\}$.
Then, for $q \in \{0,1\}$, the function
\begin{align*}
  &  m^{(q)}(w_1,\widetilde y_1,w_2,\widetilde y_2,w_3,\widetilde y_3,x,\alpha) 
    :=    -   \mathbbm{1}\left\{ \widetilde y_1 = q \right\} 
   \\ &     \qquad     
   +
        \mathbbm{1}\left\{ \widetilde y_2 = q \right\} \,        
      \exp  \left( \frac 1 2 
      \sum_{t=1}^3 \bigg\{
     \Lambda^{-1}\Big[ 
     p_{\delta(t)} \left(\widetilde y_t \, \big| \,w_{\delta(t)},x,\alpha \right)   \Big]
     -
             \Lambda^{-1}\Big[ p_t \left(\widetilde y_{t} \, \big| \,w_t,x,\alpha \right)  \Big]
             \bigg\}
         \right)
\end{align*}
satisfies
$
          \mathbb{E}\left[ m^{(q)}(W_1,\widetilde Y_1,W_2,\widetilde Y_2,W_3,\widetilde Y_3,X,A) \, \Big| \, W_1, \, X, \, A  \right] = 0.
$
\end{lemma}

The proof is given in Appendix~\ref{app:proofs}.
Note that the vector of conditioning variables $(X,A)$
is only included in Lemma~\ref{lemma:Markov} to better connect the lemma to our panel AR(1) model, but
for the mathematical result of the lemma this vector 
$(X,A)$ is actually irrelevant (no assumptions are imposed on these conditioning variables, all probability statements are conditional on $(X,A)$),
and the lemma may be easier read and understood
by initially ignoring all occurrences of $(X,A)$ and $(x,\alpha)$.
Furthermore, when applying  Lemma~\ref{lemma:Markov} to the $T=3$ panel AR(1) model of Section~\ref{sec:Derivation}, we simply have 
$(\widetilde Y_1, \widetilde  Y_2, \widetilde  Y_3)
=( Y_1,  Y_2,  Y_3)$
and 
$(W_1, W_2, W_3)
=( Y_0,  Y_1,  Y_2)$,
but the more general notation in the lemma
is convenient when generalizing the results
to models with $T>3$.

The assumptions imposed in Lemma~\ref{lemma:Markov} 
are relatively weak. In particular, 
the outcomes
 $\widetilde Y_1$, $\widetilde  Y_2$, $\widetilde  Y_3$
are not assumed to be generated from a logit model. 
However, the result of Lemma~\ref{lemma:Markov}
is in general equally weak, because the
moment functions $m^{(q)}(w_1,\widetilde y_1,w_2,\widetilde y_2,w_3,\widetilde y_3,x,\alpha)$ provided by the lemma
still depend on the individual specific effects $\alpha$,
that is, the lemma in general does {\it not} 
deliver the type of moment conditions
\eqref{MomentsConditional} that we are interested
in this paper. We find it nevertheless useful
to state the lemma in this weak form, because 
it provides some understanding for why the logit assumption
is important for obtaining valid moment conditions
that are free of the fixed effects.

For a binary choice model with single index $z_t(W_t,X) \in \mathbb{R}$ and additive fixed effects $A \in \mathbb{R}$ 
we have $\widetilde Y_t = \mathbbm{1}\{z_t(W_t,X) + A + \varepsilon_t \geq 0\}$, for  $t \in \{1,2,3\}$.
If, in addition, we assume a logistic distribution for the random shock
$\varepsilon_t$, then we obtain, for $\widetilde y \in \{0,1\}$,
\begin{align}
 p_t(\widetilde y \, |\, w_t,x,\alpha) 
   =  \Lambda\big\{ (2 \widetilde y -1) \, [ z_t(w_t,x) + \alpha ]    \big\},
    \label{GeneralLogitBinaryChoice}
\end{align}   
which implies that for all $s,t \in \{1,2,3\}$,
$$\Lambda^{-1}\Big[ 
     p_{s} \left(\widetilde y \, \big| \,w_{s},x,\alpha \right)   \Big]
     -
             \Lambda^{-1}\Big[ p_t \left(\widetilde y \, \big| \,w_t,x,\alpha \right)  \Big]
   = (2 \widetilde y -1) \, [ z_{s}(w_{s},x) - z_t(w_t,x) ]
$$
does not depend on the fixed effects $\alpha$. 
For this logistic specification with additive fixed effects
we therefore find that the moment functions
$m^{(q)}(w_1,\widetilde y_1,w_2,\widetilde y_2,w_3,\widetilde y_3,x,\alpha) $ in Lemma~\ref{lemma:Markov} do not depend on the
fixed effects $\alpha$, and can be written as\footnote{
Here we also use that
$\sum_{t=1}^3    \left[ z_{\delta(t)}(w_{\delta(t)},x) - z_t(w_t,x) \right] = 0$.
}
\begin{align}
    & m^{(q)}(w_1,\widetilde y_1,w_2,\widetilde y_2,w_3,\widetilde y_3,x) 
\nonumber \\ & \qquad
=    -   \mathbbm{1}\left\{ \widetilde y_1 = q \right\}  +
        \mathbbm{1}\left\{ \widetilde y_2 = q \right\} \,        
      \exp  \left\{  
      \sum_{t=1}^3  \widetilde y_t \, \Big[ z_{\delta(t)}(w_{\delta(t)},x) - z_t(w_t,x) \Big]
         \right\}.
   \label{MomentFunctionLogitGeneral}         
\end{align}
For the case $T=3$,
$(\widetilde Y_1, \widetilde  Y_2, \widetilde  Y_3)
=( Y_1,  Y_2,  Y_3)$,
$(W_1, W_2, W_3)
=( Y_0,  Y_1,  Y_2)$,
and $ z_{t}(w_{t},x) = y_{t-1} \, \gamma_0 + x_t' \, \beta_0$
it is easy to verify that
$m^{(q)}(w_1,\widetilde y_1,w_2,\widetilde y_2,w_3,\widetilde y_3,x)$ in
\eqref{MomentFunctionLogitGeneral} 
is equal to $m^{(q)}(y,y_0,x,\beta_0,\gamma_0)$
in display~\eqref{SolutionMomentsT3} above, that is,
Lemma~\ref{lemma:Markov}  delivers 
the moment functions derived 
for the $T=3$ dynamic logit model
in Section~\ref{sec:FindSolutions} as a special case.

\subsection{Moment conditions for $T \geq 3$}
\label{Section: T greater that 3}

We now discuss how the
moment functions for $T=3$ generalize to  more than three time periods (after the initial $y_{0}$).
We have already argued above that Lemma~\ref{lemma:Markov} is useful for our purposes for logit models of the form \eqref{GeneralLogitBinaryChoice} where it
delivers the moment functions in 
\eqref{MomentFunctionLogitGeneral} that do not depend on the fixed effects. We now apply
those results to the fixed
effect logit AR(1) model
with an arbitrary number
of time periods $T \geq 3$ by 
setting
$(\widetilde Y_1, \widetilde  Y_2, \widetilde  Y_3)
=( Y_t,  Y_s,  Y_r)$
and 
$(W_1, W_2, W_3)
=( Y_{t-1},  Y_{s-1},  Y_{r-1})$, 
for any triplet of time periods $t,s,r\in \{1,2,\ldots ,T\}$ that satisfy $t<s<r$.
Note that for this choice the Markov chain
assumption in Lemma~\ref{lemma:Markov}
is satisfied, that is, conditional on $(X,A)$,
$
   Y_{t-1}
    \rightarrow   Y_t
        \rightarrow  Y_{s-1}
        \rightarrow   Y_{s}   
       \rightarrow   Y_{r-1}
      \rightarrow    Y_r 
$
indeed constitutes a Markov chain according to model \eqref{model}.
Furthermore, in that model, the distribution
of $Y_{t}$ conditional on
$Y_{t-1}$, $X$, $A$ is indeed of 
the logistic form \eqref{GeneralLogitBinaryChoice}
with 
$z_t(w_t,x)   = 
y_{t-1} \, \gamma_0 + x_t' \, \beta_0$
for all time periods~$t$.

Making the unknown parameter dependence
explicit, 
we now define the single index for time period $t$ as $z_{t}(y,y_{0},x,\beta
,\gamma )=x_{t}^{\prime }\,\beta +y_{t-1}\,\gamma $, and we also define the
corresponding pairwise differences $z_{ts}(y,y_{0},x,\beta ,\gamma
)=z_{t}(y,y_{0},x,\beta ,\gamma )-z_{s}(y,y_{0},x,\beta ,\gamma )$. Then,
for triples of time periods $t,s,r\in \{1,2,\ldots ,T\}$ with $t<s<r$,
the moment function in \eqref{MomentFunctionLogitGeneral} 
can be written more explicitly as
\begin{align}
m^{(0)(t,s,r)}(y,y_0,x,\beta ,\gamma )& =\left\{
\begin{array}{ll}
\exp \left[ z_{sr}(y,y_{0},x,\beta ,\gamma )\right] -1 & \text{if }%
(y_{t},y_{s},y_{r})=(0,0,1), \\
-1 & \text{if }(y_{t},y_{s})=(0,1), \\
\exp \left[ z_{rt}(y,y_{0},x,\beta ,\gamma )\right]  & \text{if }%
(y_{t},y_{s},y_{r})=(1,0,0), \\
\exp \left[ z_{st}(y,y_{0},x,\beta ,\gamma )\right]  & \text{if }%
(y_{t},y_{s},y_{r})=(1,0,1), \\
0 & \text{otherwise},%
\end{array}%
\right.
\nonumber \\[5pt]
m^{(1)(t,s,r)}(y,y_0,x,\beta ,\gamma )& =\left\{
\begin{array}{ll}
\exp \left[ z_{ts}(y,y_{0},x,\beta ,\gamma )\right]  & \text{if }%
(y_{t},y_{s},y_{r})=(0,1,0), \\
\exp \left[ z_{tr}(y,y_{0},x,\beta ,\gamma )\right]  & \text{if }%
(y_{t},y_{s},y_{r})=(0,1,1), \\
-1 & \text{if }(y_{t},y_{s})=(1,0), \\
\exp \left[ z_{rs}(y,y_{0},x,\beta ,\gamma )\right] -1 & \text{if }%
(y_{t},y_{s},y_{r})=(1,1,0), \\
0 & \text{otherwise}.%
\end{array}%
\right.   
\label{MomentsGeneralT}
\end{align}%
For $T=3$ and $(t,s,r)=(1,2,3)$, these moment functions are exactly those
calculated in Section~\ref{sec:FindSolutions} above.
For general $T \geq 3$
and triplets $(t,s,r)$
we can apply Lemma~\ref{lemma:Markov}, conditional also on $Y_0,Y_{1},\ldots ,Y_{t-1}$,
to  obtain the following theorem.
\begin{theorem}
\label{th:AR1moments} If the outcomes $Y$ are generated from the panel
logit AR(1) model with $T\geq 3$ and true parameters $\beta _{0}$ and $%
\gamma _{0}$, then we have for all $t,s,r\in \{1,2,\ldots ,T\}$ with $t<s<r$%
, and for all $q \in \{0,1\}$, $y^{(t)} \in \{0,1\}^t$, $x\in \mathbb{R}^{K\times T}$,  
  $\alpha \in \mathbb{R}$,    that
\begin{align*}
\mathbb{E}\left[  m^{(q)(t,s,r)}(Y,Y_0,X,\beta
_{0},\gamma _{0})\,\big|\,(Y_{0},\, Y_{1},\ldots ,Y_{t-1})=y^{(t)}, \,X=x,\,A=\alpha \right] & =0.
\end{align*}

\end{theorem}
The proof is given in Appendix~\ref{app:proofs}, but as argued above, the theorem
 really is an immediate corollary of Lemma~\ref{lemma:Markov}.
Instead of conditioning on $Y_{1},\ldots ,Y_{t-1}$, we can also multiply 
the moment function with an arbitrary 
function of $Y_{1},\ldots ,Y_{t-1}$. Namely, by applying Theorem~\ref{th:AR1moments}
and the law of iterated expectations, we find, for
any function $w:\{0,1\}^{t-1}\rightarrow \mathbb{R}$, that
\begin{align}
\mathbb{E}\left[ w(Y_{1},\ldots ,Y_{t-1}) \; m^{(q)(t,s,r)}(Y,Y_0,X,\beta
_{0},\gamma _{0})\,\big|\,Y_{0}=y_{0},\,X=x,\,A=\alpha \right] & =0.
   \label{allAR1moments}
\end{align}
From Section~\ref{SEC: On the number of moment conditions} we know that,
for any fixed value of the initial condition $y_0$, 
there are at least $\ell=2^T-2T$ linearly independent moment conditions
available for our AR(1) logit model with $T$ time periods.
It turns out for $\gamma \neq 0$ this is exactly the
correct number of linearly independent moment conditions in this model. 
 In the 2020 working paper version of the current paper we conjectured this,
and subsequent papers by
\cite{kruiniger2020further},
\cite{dobronyi2021identification},
and \cite{Dano2023arXiv}
have shown that this is indeed the case.

Equation \eqref{allAR1moments}
provides all of the $\ell=2^T-2T$ available valid moment
functions for this model, but not all those moment functions 
$w(Y_{1},\ldots ,Y_{t-1})\,m^{(q)(t,s,r)}(Y,Y_0,X,\beta
_{0},\gamma)$ are  linearly independent,
that is, some of them
can be written
as linear combinations
(with coefficients
that depend on $x$, $\beta$, $\gamma$) of the others.
However, if we restrict ourselves to $r=T$, then we
 have verified numerically that
a linearly independent basis is obtained. 
Note that once we fix $r=T$, then, for given $y_0$, 
we can still choose $q \in \{0,1\}$,
$w:\{0,1\}^{t-1}\rightarrow \mathbb{R}$,
and $(t,s)$, with $1\leq t<s<T$. The total number of basis elements
is therefore equal to
\begin{align*}
   \ell =  \underbrace{ 2 }_{\text{$q \in \{0,1\}$}}  \times
     \underbrace{ \sum_{t=1}^{T-2} \sum_{s=t+1}^{T-1} }_{\text{allowed $t,s$ values}}
    \underbrace{ 2^{t-1} }_{\begin{minipage}{3.2cm} \center \scriptsize
    \setstretch{1.0}
    number of linearly independent functions $w(y_{1},\ldots ,y_{t-1})$
      \end{minipage}}
    =   \sum_{t=1}^{T-2} \, 2^t \, (T-t-1) = 2^T - 2T,
\end{align*}
as claimed above.\footnote{%
We consider $\gamma \neq 0$ here. For $\gamma=0$ we have a static panel logit model, and in that case $T-1$ additional moment conditions
become available, bringing the total number of available moments to $2^T - T -1$. 
The first-order conditions
of the conditional likelihood in \cite{rasch1960studies} and \cite{andersen1970asymptotic} are linear combinations of these moment functions.}

\subsubsection{Unbalanced panels and missing time periods}
The only regressor and outcome values that enter into
the moment functions $m^{(q)(t,s,r)} $
are $(x_t,x_s,x_r)$ and $(y_{t-1},y_t,y_{s-1},y_s,y_{r-1},y_t)$.\footnote{Of course, $y_t$ coincides with $y_{s-1}$ if $t=s-1$,
and $y_s$ coincides with $y_{r-1}$ if $s=r-1$.}
Thus, as long as those variables are observed we can evaluate $m^{(q)(t,s,r)} $.
The moment conditions for $T>3$ can therefore also be applied to unbalanced panels where regressors and outcomes
are not observed in all time periods,  provided that the
occurrence of missing values
is independent of the outcomes $Y$, conditional on the regressors $X$ and the individual-specific effects $A$. The data in our empirical illustration are indeed unbalanced, and  in
 Section~\ref{sec:Emp} we discuss how to combine the moment functions for unbalanced panels.

\subsubsection{Relation to \cite{kitazawa2013exploration,kitazawa2016root}}

The first paper to obtain moment conditions for the dynamic
panel logit model without imposing restrictions on the
covariate values is the working paper by
\cite{kitazawa2013exploration}, which was recently published
(\citealt{kitazawa2022transformations}).
That paper defines
\begin{align}
    U_t &= y_t + (1-y_t) y_{t+1} - (1-y_t) y_{t+1} \exp(-\beta \Delta x_{t+1}) - \delta y_{t-1} (1-y_t)y_{t+1}  \exp(-\beta \Delta x_{t+1}),
  \nonumber \\
    \hbar U_t &= U_t - y_{t-1} - \tanh\left[ \frac{ -\gamma y_{t-2} + \beta (\Delta x_t + \Delta x_{t+1}) } 2 \right]  
    (U_t + y_{t-1} - 2U_t y_{t-1}) ,
 \nonumber  \\
    \Upsilon_t  &= y_t y_{t+1} + y_t (1-y_{t+1}) \exp(\beta \Delta x_{t+1}) + \delta (1-y_{t-1}) y_t (1-y_{t+1}) \exp(\beta \Delta x_{t+1}) ,
 \nonumber  \\
    \hbar   \Upsilon_t  &= \Upsilon_t  - y_{t-1} -  \tanh\left[ \frac{ \gamma (1-y_{t-2}) + \beta (\Delta x_t + \Delta x_{t+1}) } 2 \right]
     (\Upsilon_t + y_{t-1} - 2\Upsilon_t y_{t-1}) ,
     \label{KitazawaMomentFunctions}
\end{align}
where $\delta=e^\gamma -1$ and $\Delta x_t=x_t - x_{t-1}$.
The paper then shows  that, for $t \in \{2,\ldots,T-1\}$,\footnote{This is written here in our conventions for $t$ and $T$.
}
the functions $ \hbar U_t$ and $  \hbar   \Upsilon_t $ are valid moment functions, in the sense of \eqref{MomentsConditional}.
\cite{kitazawa2016root} uses the same moment conditions, but also includes time dummies in the model, which in our notation
are included in the parameter vector $\beta$ (one just needs to define the regressors $x_t$ as appropriate dummy variables).

Those definitions  look quite different to our moment functions above, but one can show that
\begin{align*}
    \hbar U_2 &=  \left\{  \tanh\left[ \frac{ -\gamma y_{0} + \beta (\Delta x_2 + \Delta x_{3}) } 2 \right]  -1  \right\} \,  m^{(0)}(y,y_0,x,\beta ,\gamma ) ,
    \\
    \hbar  \Upsilon_2 &=  \left\{  \tanh\left[ \frac{ \gamma (1-y_{0}) + \beta (\Delta x_2 + \Delta x_{3}) } 2 \right]   +1  \right\} \,  m^{(1)}(y,y_0,x,\beta ,\gamma ) .
\end{align*}
Thus, apart from a  rescaling (with a non-zero function of the parameters and conditioning variables), the moment functions of
\cite{kitazawa2013exploration} coincide with our moment functions for AR(1) models with $T=3$.
However, the complete set of  
moment conditions for $T>3$
in Theorem~\ref{th:AR1moments}  is new.

\subsubsection{Relation to other existing results}

\cite{honore2000panel} observe that with \eqref{model} and $T=3$, the
conditional likelihood function that conditions on $Y=y_{0}$, $Y_{3}=y_{3}$,
and $Y_{1}+Y_{2}=1$,
\begin{equation*}
\ell _{y_{0},y_{3}}(y,x,\beta ,\gamma )= \,\mathrm{Pr}\left( Y=y\,\big|%
\,Y_{0}=y_{0},\,Y_{1}+Y_{2}=1,\,Y_{3}=y_{3},\,X=x,\beta ,\gamma \right) ,
\end{equation*}
does not depend on $\alpha $, when $x=(x_{1},x_{2},x_{2})$ (so the
explanatory variables are the same in the last two periods). The
corresponding scores are
\begin{align}
\frac{\partial \ell _{0,0}(y,x,\beta ,\gamma )}{\partial \gamma }& =0,
\notag \\
\frac{\partial \ell _{0,0}(y,x,\beta ,\gamma )}{\partial \beta }& =\frac{%
x_{12}}{1+\exp \left( x_{12}^{\prime }\beta \right) }\left[ \frac{%
m^{(1)}(y,0,x,\beta ,\gamma )+\exp \left( x_{12}^{\prime }\beta -\gamma
\right) \,m^{(0)}(y,0,x,\beta ,\gamma )}{\exp \left( -\gamma \right) -1}%
\right] ,  \label{FOClikelihood1}
\end{align}%
and
\begin{align}
\left(
\begin{array}{c}
\displaystyle\frac{\partial \ell _{0,1}(y,x,\beta ,\gamma )}{\partial \gamma
} \\[3pt]
\displaystyle\frac{\partial \ell _{0,1}(y,x,\beta ,\gamma )}{\partial \beta }%
\end{array}%
\right) & =\left(
\begin{array}{c}
\displaystyle-1\nonumber \\[3pt]
\displaystyle x_{12}%
\end{array}%
\right) \frac{1}{1+\exp \left( x_{12}^{\prime }\beta -\gamma \right) }
\notag \\
& \qquad \qquad \times \left[ \frac{m^{(1)}(y,0,x,\beta ,\gamma )+\exp
\left( x_{12}^{\prime }\beta \right) \,m^{(0)}(y,0,x,\beta ,\gamma )}{\exp
(\gamma )-1}\right] ,  \label{FOClikelihood2}
\end{align}%
where $m^{(0)}$ and $m^{(1)}$ are defined in {\eqref{SolutionMomentsT3}. The results for $y_{0}=1$
are analogous. Thus, the score functions of the conditional likelihood in
\cite{honore2000panel} are linear combinations of our moment conditions when
$x_{2}=x_{3}$. The conditional likelihood estimation discussed in \cite%
{cox1958regression} and \cite{Chamberlain1985} are special cases of this
without regressors ($x_{1}=x_{2}=x_{3}=0$).

\cite{Hahn2001} considers model \eqref{model} with $T=3$, initial condition $%
y_{0}=0$, and time dummies as regressors, that is, $x_{t}^{\prime }\beta
=\beta _{t}$, with the normalization $\beta _{1}=0$. The common parameters
in that model are $(\beta _{2},\beta _{3},\gamma )$. Hahn shows that  these
parameters cannot be estimated at root-n-rate. This is not in conflict
with our results here, because Lemma~\ref{lemma:moments_p1T3} only provides
two moment conditions for $y_{0}=0$. However, there are three model
parameters in the setup of \cite{Hahn2001}, so just from counting parameters
and moment conditions, we know that our moments cannot identify $(\beta _{2},\beta _{3},\gamma )$.%
\footnote{%
Including moment conditions that use the initial condition $y_{0}=1$ will
give two additional moments. From the point of view of counting moments,
this will result in a model which is over-identified.} Thus, our moment
conditions cannot be used to estimate the parameters $(\beta _{2},\beta
_{3},\gamma )$ at root-n-rate. This is in agreement with Hahn's
calculation of the information bound for this model.

The main reason why we can identify and estimate $\beta $ and $\gamma $ is
that we consider non-constant regressors $X=(X_{1},X_{2},X_{3})$, which
gives us two moment conditions for each initial condition and each support
point of the regressors, and thus many more moment conditions
than parameters --- see our formal
results on point-identification of $\beta $ and $\gamma $ 
in Section~\ref{sec:identitifactionAR1} above.

\subsection{More general dynamic panel models}

As explained in Section~\ref{subsec:MasterLemma},
Lemma~\ref{lemma:Markov}
delivers valid moment functions
for any model with logistic
conditional
probabilities of the form
\eqref{GeneralLogitBinaryChoice}.
This means that 
the single index of the model
need not be of the form 
$x_{t}^{\prime }\,\beta +y_{t-1}\,\gamma$
that is linear in $x_{t}$
and $y_{t-1}$, but it can
actually be any function
of the strictly exogenous regressors,
lagged dependent variable,
and parameters. 
In particular, if we replace
the model specification 
\eqref{model} by
\begin{equation}
\mathrm{Pr}\left( Y_{it}=1\,\big|\,Y_{i}^{t-1},X_{i},A_{i}\right) =\frac{%
\exp \big[
(1-Y_{i,t-1}) X_{it}^{\prime }\,\beta_0
+
Y_{i,t-1} X_{it}^{\prime }\,\beta_1 +Y_{i,t-1}\,\gamma +A_{i}\big)}{1+\exp %
\big((1-Y_{i,t-1}) X_{it}^{\prime }\,\beta_0
+
Y_{i,t-1} X_{it}^{\prime }\,\beta_1 +Y_{i,t-1}\,\gamma +A_{i}\big]},  \label{modelGENERALIZED}
\end{equation}%
then the moment functions \eqref{MomentsGeneralT}
and Theorem~\ref{th:AR1moments}
remain fully valid, as long as
we replace the parameters
$(\beta,\gamma)$ by $(\beta_0,\beta_1,\gamma)$,
and define the single index
by 
$z_{t}(y,y_{0},x,\beta_0,\beta_1,
,\gamma )=(1-y_{t-1}) \,x_{t}^{\prime }\,\beta_0 + y_{t-1} \,x_{t}^{\prime }\,\beta_1 +y_{t-1}\,\gamma $.

The generalized model \eqref{modelGENERALIZED} is interesting, because it allows the effect of the regressors $X_{it}$ on $Y_{it}$ to 
depend on the current ``state'' of the
process, $Y_{i,t-1}$, with 
$\beta_{0/1}$ measuring the effect
of $X_{it}$ on $Y_{it}$ if 
$Y_{i,t-1} = 0/1$.
We do not consider this more
general model structure further in this paper,
but it is noteworthy that the
regressors $(1-Y_{i,t-1}) X_{it}$
and $Y_{i,t-1} X_{it}$ are 
pre-determined
regressors that are more general 
than just the lagged dependent
variable $Y_{i,t-1}$ we have considered
so far. 
Further comments on more general
regressors structures are given in Appendix~\ref{app:MoreGeneralPredetermined}.

Another interesting generalization of this model that still allows for the construction
of moment conditions, is to make the AR(1) coefficient 
$\gamma$ individual specific. Equation \eqref{modelGENERALIZED} then reads
\begin{align}
&\mathrm{Pr}\left( Y_{it}=1\,\big|\,Y_{i}^{t-1},X_{i},A_{i},C_i\right) 
\nonumber \\ & \qquad \qquad  \qquad
=\frac{%
\exp \big[(1-Y_{i,t-1}) X_{it}^{\prime }\,\beta_0
+
Y_{i,t-1} X_{it}^{\prime }\,\beta_1 +Y_{i,t-1}\,C_i +A_{i}\big]}{1+\exp %
\big[(1-Y_{i,t-1}) X_{it}^{\prime }\,\beta_0
+
Y_{i,t-1} X_{it}^{\prime }\,\beta_1 +Y_{i,t-1}\,C_i +A_{i}\big]}.
   \label{modelGammaHeterogeneous}
\end{align}
We can then treat $(C_i,A_i)$ as a two-dimensional fixed-effect and employ the
methods of Section~\ref{SEC: Incidental parameter free moment conditions} and \ref{sec:Derivation} to explore moment conditions for $\beta$ that are free of
$(C_i,A_i)$. For general covariate and parameter values, 
we find that no such moment conditions exist for $T=3$, but they do exist for $T \geq 4$. For example, for $T=4$ and $y_0=0$, a valid moment condition in this model
(i.e.\ satisfying $\mathbb{E}\left[ m(Y_i,Y_{0,i},X_i,\beta_0,\beta_1 )\big|%
\,Y_{0,i}=0,X_i,C_i,A_i\right] =0$) is
\begin{align}
m(y,0,x,\beta_0,\beta_1 )& =
- \mathbbm{1}\Big\{ (y_1,y_3,y_4) = (1,0,0) \Big\}
+
\left\{
\begin{array}{ll}
\exp \left( z_{12}   \right)   & \text{if }y=(0,1,0,0), \\
\exp \left( z_{14}   \right)   & \text{if }y=(0,1,0,1), \\
- \exp \left( z_{34}   \right)   & \text{if }y=(1,0,0,1), \\
\exp \left( z_{32}  \right)   & \text{if }y=(1,1,0,0), \\
0 & \text{otherwise},%
\end{array}
\right.
   \label{momentGammaHeterogeneous}
\end{align}
where $z_{ts} = z_t - z_s$ as before,
and $z_t=z_{t}(y,y_{0},x,\beta_0 ,\beta_1 )
 = (1-y_{t-1}) x_{t}^{\prime }\,\beta_0
+
y_{t-1} x_{t}^{\prime }\,\beta_1$
is the appropriate index function in this model.
We have found numerically that there is one additional moment condition for $T=4$,
a total of ten for $T=5$, and thirty-two for $T=6$.
The moment function in the last display was derived using the ideas described in
Section~\ref{sec:Derivation}. 

\section{Empirical illustration\label{SEC: Empirical illustration}}

\label{sec:Emp}

In this section, we illustrate how to use the conditional moment functions
in this paper to implement a GMM\ approach to estimation. We use data from
the National Longitudinal Survey of Youth 1997\footnote{%
The analysis is restricted to the years in which the survey was conducted
annually, from 1997-2011. For years in which the respondent was not
interviewed, all time-varying variables (e.g., employment status, school
enrollment status, age, income, marital variables, etc) are marked as
missing. Otherwise, unless the raw data was marked as missing in some
capacity (e.g., due to non-response, the interviewee not knowing the answer
to the question), no other entries had missings imposed upon them.} (NLSY97)
covering the years 1997 to 2010, and the dependent variable is a binary
variable indicating employment status by whether the respondent reported
working $\geq $1000 hours in the past year. We estimate fixed effects logit
AR(1) and AR(2) models using the number of biological children the
respondent has (Children), a dummy variable for being married (Married), a
transformation\footnote{%
The spouse's income can be zero or negative. This prevents us from using the
logarithm of the income as an explanatory variable. We therefore use the
signed fourth root.} of the spouse's income (Sp.Inc.), and a full set of
time dummies as the explanatory variables. There are a total of 8,274
individuals aged 16 to 32, resulting in 54,166 observations. For the
estimation, we consider the full sample, as well as females and males
separately. Figure \ref{Figure: Histograms} displays the number of
observations, $T_{i}$, per individual in each of the three samples.

\begin{figure}[h]
\caption{Histogram of Number of Observations Per Individual.}
\label{Figure: Histograms}\centering
\begin{subfigure}[t]{0.30\textwidth}
        \includegraphics[scale=0.35]{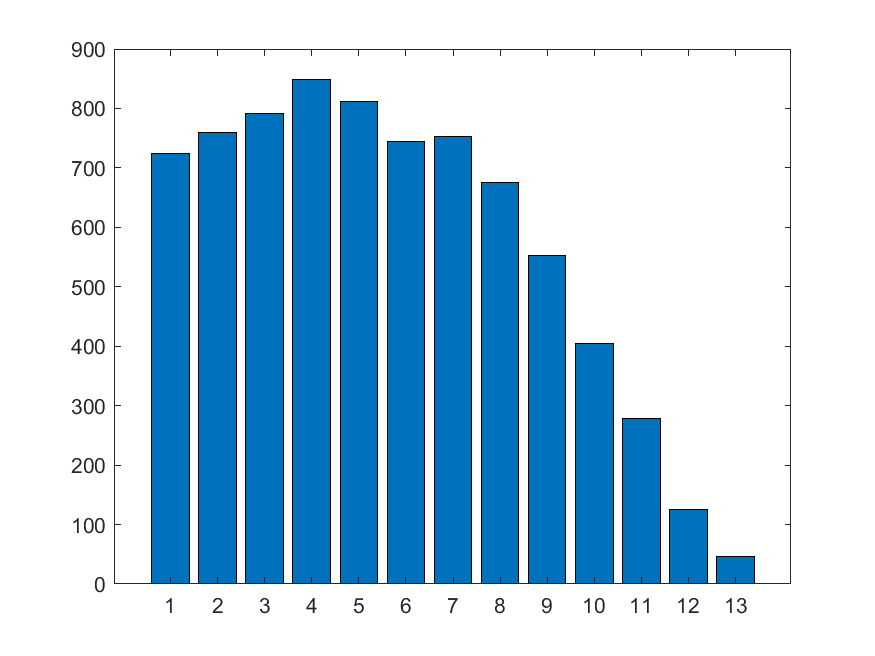}
        \caption{All}
    \end{subfigure} \quad
\begin{subfigure}[t]{0.30\textwidth}
        \includegraphics[scale=0.35]{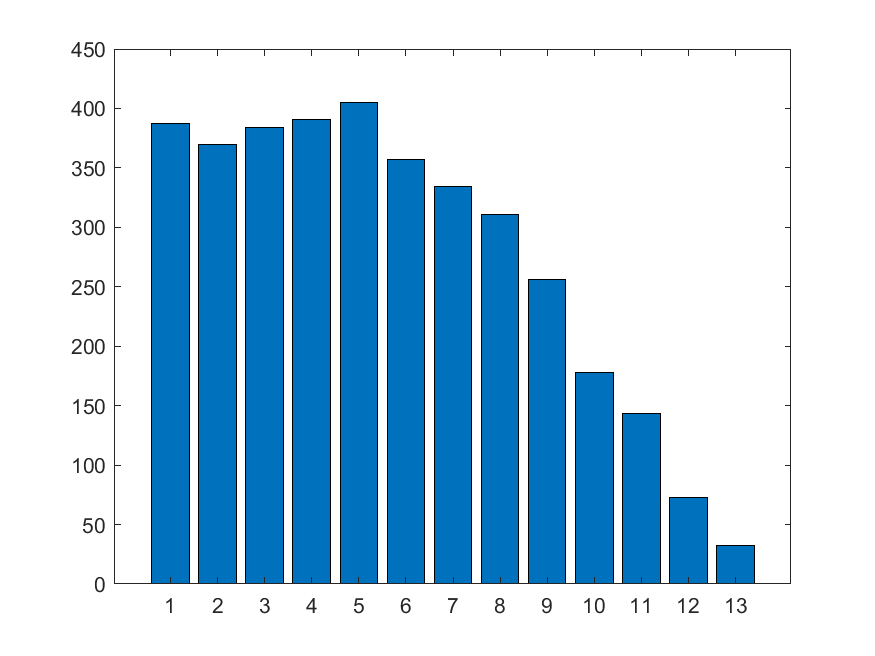}
        \caption{Females}
    \end{subfigure} \quad
\begin{subfigure}[t]{0.30\textwidth}
        \includegraphics[scale=0.35]{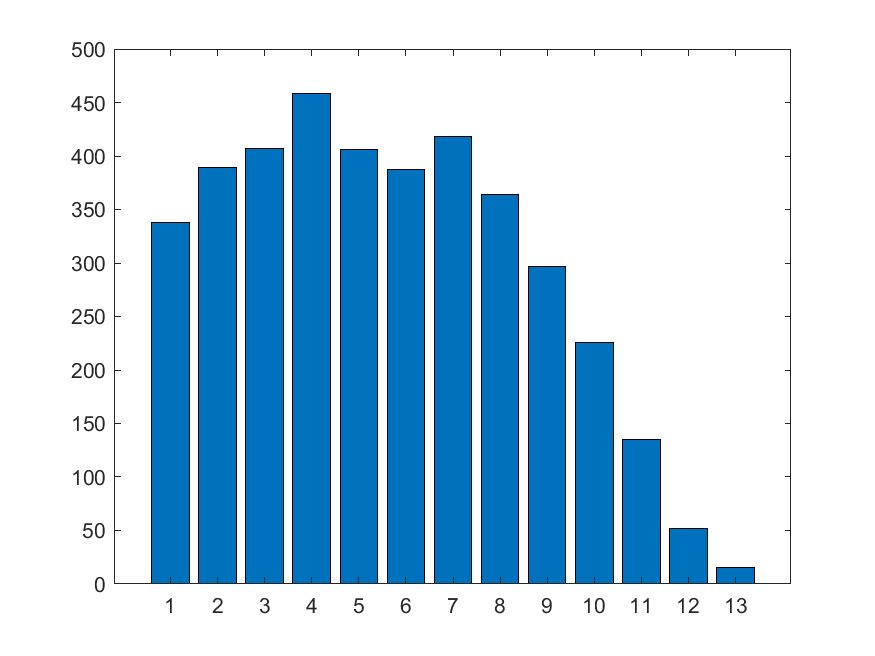}
        \caption{Males}
    \end{subfigure}
\end{figure}

The moment conditions for the fixed effects logit AR(1) in (\ref%
{SolutionMomentsT3}) and (\ref{MomentsGeneralT}) are all indexed by three
time-periods, and they are conditional on the strictly exogenous variables
and the initial conditions. One could in principle construct separate
conditional moment conditions for each value of $T_{i}$ and each triplet $%
1\leq t<s<r\leq T_{i}$, and then use them to construct efficient
unconditional moment functions. See, for example, the discussion in \cite{NeweyMcFadden94:HoE}. Unfortunately, the construction of these moment
functions depends on the conditional expectation of the derivative of the
conditional moment function as well as on the conditional variance of the
conditional moment function. We therefore pursue a different approach to
obtaining unconditional moment functions. We do not claim that the resulting
GMM estimator has any optimality properties, but we have found that it
performs well in our Monte Carlo simulations even for relatively small
sample sizes, see Appendix~B.1 of \cite{Honore2022moment}.

We first normalize all moment functions such that $\sup_{y,x,\beta ,\gamma
}\allowbreak \left\vert \widetilde{m}(y,x,\beta ,\gamma )\right\vert <\infty
$. For example, the rescaled versions of our $T=3$ moment functions in
Section~\ref{sec:FindSolutions} are given by
\begin{eqnarray*}
\widetilde{m}^{(0)}(y,y_{0},x,\beta ,\gamma ) &=&\frac{m^{(0)}(y,y_{0},x,%
\beta ,\gamma )}{1+\exp \left( x_{23}^{\prime }\beta \right) +\exp \left(
x_{31}^{\prime }\beta -y_{0}\,\gamma \right) +\exp \left( x_{21}^{\prime
}\beta +(1-y_{0})\,\gamma \right) } \\
\widetilde{m}^{(1)}(y,y_{0},x,\beta ,\gamma ) &=&\frac{m^{(2)}(y,y_{0},x,%
\beta ,\gamma )}{1+\exp \left( x_{12}^{\prime }\beta +y_{0}\,\gamma \right)
+\exp \left( x_{13}^{\prime }\beta -(1-y_{0})\,\gamma \right) +\exp \left(
x_{32}^{\prime }\beta \right) }
\end{eqnarray*}%
Here, each moment function is divided by the sum of the absolute values of
all the different positive summands that appear in that moment function. We
have found that this rescaling improves the performance of the resulting GMM
estimators, particularly for small samples, because it bounds the moment
functions and its gradients uniformly over the parameters $\beta $ and $%
\gamma $. Interestingly, the score functions of the conditional likelihood
in \cite{honore2000panel} are essentially rescaled in this way.

The rescaled moment functions are valid conditional on any realization of
the regressors. We can therefore form unconditional moment functions by
multiplying them with arbitrary functions of the regressors and the initial
conditions. In our example, we multiply them by 1, the initial condition,
and the explanatory variables for the three time periods that index the
moment function. For example, for $T=3$, we use
\begin{equation*}
M(y,y_{0},x,\beta ,\gamma )=\big(1,\;y_{0},\;x_{1}^{\prime },\;x_{2}^{\prime
},\;x_{3}^{\prime }\big)^{\prime }\otimes \left(
\begin{array}{c}
\widetilde{m}^{(1)}(y,y_{0},x,\beta ,\gamma ) \\
\widetilde{m}^{(0)}(y,y_{0},x,\beta ,\gamma )%
\end{array}%
\right) ,
\end{equation*}%
where $\otimes $ is the tensor product.

One could in principle construct a moment function for each $\left(
t,s,r\right) $ which indexes a moment function. However, this would create a
very large number of moment conditions. For a given individual, we therefore
add up all the moment functions over all triplets, $t<s<r$. Observations
with $T=T_{i}$ time periods will then contribute ${\binom{T_{i}}{3}}$ terms
to the sample analog of the moment. This gives very large weight to
observations with large $T_{i}$. 
We therefore weigh the triplets $%
\left( t,s,r\right) $ for an observation with $T_{i}$ time periods by $%
\left( T_{i}-1\right) /{\binom{T_{i}}{3}}$. This yields sample moments of
the form $\frac{1}{n}\sum_{i=1}^{n}M(Y_{i},Y_{i,0},X_{i},\beta ,\gamma )$,
and the corresponding GMM estimator is given by
\begin{equation*}
{\binom{\widehat{\beta }}{\widehat{\gamma }}}=\limfunc{argmin}_{\beta \in
\mathbb{R}^{K},\,\gamma \in \mathbb{R}}\left(
\sum_{i=1}^{n}M(Y_{i},Y_{i,0},X_{i},\beta ,\gamma )\right) ^{\prime }W\left(
\sum_{i=1}^{n}M(Y_{i},Y_{i,0},X_{i},\beta ,\gamma )\right) ,
\end{equation*}%
where $W$ is a symmetric positive-definite weight matrix. We use a diagonal
weight matrix with the inverse of the moment variances on the diagonal. The
motivation stems from \cite{AltonjiSegal1996} who demonstrate that
estimating the optimal weighting matrix can result in poor finite sample
performance of GMM estimators. They suggest equally weighted moments (i.e., $%
W=I$) as an alternative. Of course, using equal weights will not be
invariant to changes in units, which explains the practice we have adopted.%
\footnote{%
Our choice of weight matrix is quite common in empirical work. See, for
example, \cite{GayleShephard2019} for a recent example.}

Under standard regularity conditions we have
\begin{equation*} 
\sqrt{n}\left[ {\binom{\widehat{\beta }}{\widehat{\gamma }}}-{\binom{\beta
_{0}}{\gamma _{0}}}\right] \Rightarrow \mathcal{N}\left(0, \,
(G'WG)^{-1} \,
G^{\prime }W\,\Omega \,WG\,
(G'WG)^{-1}
\right),
\end{equation*}%
with $\Omega =\mathrm{Var}[m(Y_{i},Y_{i,0},X_{i},\beta _{0},\gamma _{0})]$
and $G=\mathbb{E}\left[ \frac{\partial m(Y_{i},Y_{i,0},X_{i},\beta
_{0},\gamma _{0})}{\partial \beta ^{\prime }},\frac{\partial
m(Y_{i},Y_{i,0},X_{i},\beta _{0},\gamma _{0})}{\partial \gamma ^{\prime }}%
\right] $.

Table \ref{Table: Empirical} reports the estimation results. As expected,
and consistent with the Monte Carlo results  in Appendix~B.1 of \cite{Honore2022moment}, the standard logit maximum
likelihood estimator of the coefficient on the lagged dependent variable is
much larger than the one that estimates a fixed effect for each individual:
the estimated fixed effects will be \textquotedblleft
`overfitted\textquotedblright , leading to a downward bias in the estimated
state dependence. Moreover, the standard logit estimator  that ignores fixed effects will capture the
presence of persistent heterogeneity by the lagged dependent variable,
leading to an upwards bias if such heterogeneity is present in the data. The
GMM\ estimator gives a much smaller coefficient than the standard logit
maximum likelihood estimator, suggesting that heterogeneity plays a big role
in this application.

\begin{table}[tbp!]
\caption{Empirical Results (AR(1)).}
\label{Table: Empirical}\centering{\hspace{-0.3cm} {\footnotesize
\begin{tabular}{l@{\;\;\;\;}r@{\;\;\;\;}r@{\;\;\;\;}r@{\;\;\;\;}r@{\;\;\;\;}r@{\;\;\;\;}r@{\;\;\;\;}r@{\;\;\;\;}r@{\;\;\;\;}r@{\;\;\;\;}r@{\;\;\;\;}r}
&  &  &  &  &  &  &  &  &  &  &  \\
& \multicolumn{3}{c}{Females} &  & \multicolumn{3}{c}{Males} &  &
\multicolumn{3}{c}{All} \\
&  &  &  &  &  &  &  &  &  &  &  \\
& \multicolumn{1}{c}{Logit} & \multicolumn{1}{c}{%
\begin{tabular}{@{}r}
Logit \\
w FE%
\end{tabular}%
} & \multicolumn{1}{c}{GMM} &  & \multicolumn{1}{c}{Logit} &
\multicolumn{1}{c}{%
\begin{tabular}{@{}r}
Logit \\
w FE%
\end{tabular}%
} & \multicolumn{1}{c}{GMM} &  & \multicolumn{1}{c}{Logit} &
\multicolumn{1}{c}{%
\begin{tabular}{@{}r}
Logit \\
w FE%
\end{tabular}%
} & \multicolumn{1}{c}{GMM} \\ \cline{2-4}\cline{6-8}\cline{10-12}
&  &  &  &  &  &  &  &  &  &  &  \\
Lagged $y$ & $2.585 $ & $0.780 $ & $1.512 $ &  & $2.947 $ & $0.709 $ & $%
1.454 $ &  & $2.797 $ & $0.768 $ & $1.417 $ \\
& $( 0.038 \rlap{)}$ & $( 0.050 \rlap{)}$ & $( 0.076 \rlap{)}$ &  & $( 0.040 %
\rlap{)}$ & $( 0.063 \rlap{)}$ & $( 0.088 \rlap{)}$ &  & $( 0.027 \rlap{)}$
& $( 0.039 \rlap{)}$ & $( 0.060 \rlap{)}$ \\
&  &  &  &  &  &  &  &  &  &  &  \\
Children & $-0.335 $ & $-0.444 $ & $-0.244 $ &  & $-0.153 $ & $0.018 $ & $%
-0.275 $ &  & $-0.278 $ & $-0.252 $ & $-0.214 $ \\
& $( 0.016 \rlap{)}$ & $( 0.052 \rlap{)}$ & $( 0.196 \rlap{)}$ &  & $( 0.021 %
\rlap{)}$ & $( 0.067 \rlap{)}$ & $( 0.133 \rlap{)}$ &  & $( 0.012 \rlap{)}$
& $( 0.043 \rlap{)}$ & $( 0.102 \rlap{)}$ \\
&  &  &  &  &  &  &  &  &  &  &  \\
Married & $0.082 $ & $-0.044 $ & $0.637 $ &  & $0.335 $ & $0.332 $ & $0.038 $
&  & $0.349 $ & $0.173 $ & $0.707 $ \\
& $( 0.084 \rlap{)}$ & $( 0.159 \rlap{)}$ & $( 0.890 \rlap{)}$ &  & $( 0.071 %
\rlap{)}$ & $( 0.171 \rlap{)}$ & $( 0.295 \rlap{)}$ &  & $( 0.053 \rlap{)}$
& $( 0.111 \rlap{)}$ & $( 0.397 \rlap{)}$ \\
&  &  &  &  &  &  &  &  &  &  &  \\
SP.Inc. & $-0.010 $ & $-0.050 $ & $-0.104 $ &  & $0.033 $ & $0.003 $ & $%
0.019 $ &  & $-0.017 $ & $-0.044 $ & $-0.089 $ \\
& $( 0.006 \rlap{)}$ & $( 0.011 \rlap{)}$ & $( 0.068 \rlap{)}$ &  & $( 0.007 %
\rlap{)}$ & $( 0.016 \rlap{)}$ & $( 0.026 \rlap{)}$ &  & $( 0.004 \rlap{)}$
& $( 0.009 \rlap{)}$ & $( 0.033 \rlap{)}$ \\
&  &  &  &  &  &  &  &  &  &  &  \\
&  &  &  &  &  &  &  &  &  &  &
\end{tabular}
}}
\begin{tablenotes}
      \small
      \item The estimation also includes 12 time dummies. Standard error for the GMM and Logit Fixed Effects Estimators are calculated as the interquartile range of 1,000 bootstrap replications divided by 1.35.
\end{tablenotes}
\end{table}

\begin{table}[tbp!]
\caption{Empirical Results (AR(2)).}
\label{Table: Empirical AR2}\centering{\hspace{-0.3cm} {\footnotesize
\begin{tabular}{l@{\;\;\;\;}r@{\;\;\;\;}r@{\;\;\;\;}r@{\;\;\;\;}r@{\;\;\;\;}r@{\;\;\;\;}r@{\;\;\;\;}r@{\;\;\;\;}r@{\;\;\;\;}r@{\;\;\;\;}r@{\;\;\;\;}r}
&  &  &  &  &  &  &  &  &  &  &  \\
& \multicolumn{3}{c}{Females} &  & \multicolumn{3}{c}{Males} &  &
\multicolumn{3}{c}{All} \\
&  &  &  &  &  &  &  &  &  &  &  \\
& \multicolumn{1}{c}{Logit} & \multicolumn{1}{c}{%
\begin{tabular}{@{}r}
Logit \\
w FE%
\end{tabular}%
} & \multicolumn{1}{c}{GMM} &  & \multicolumn{1}{c}{Logit} &
\multicolumn{1}{c}{%
\begin{tabular}{@{}r}
Logit \\
w FE%
\end{tabular}%
} & \multicolumn{1}{c}{GMM} &  & \multicolumn{1}{c}{Logit} &
\multicolumn{1}{c}{%
\begin{tabular}{@{}r}
Logit \\
w FE%
\end{tabular}%
} & \multicolumn{1}{c}{GMM} \\ \cline{2-4}\cline{6-8}\cline{10-12}
&  &  &  &  &  &  &  &  &  &  &  \\
$y_{t-1}$ & $2.259 $ & $0.742 $ & $1.356 $ &  & $2.422 $ & $0.514 $ & $1.116
$ &  & $2.361 $ & $0.665 $ & $1.297 $ \\
& $( 0.047 \rlap{)}$ & $( 0.069 \rlap{)}$ & $( 0.162 \rlap{)}$ &  & $( 0.052 %
\rlap{)}$ & $( 0.083 \rlap{)}$ & $( 0.131 \rlap{)}$ &  & $( 0.035 \rlap{)}$
& $( 0.053 \rlap{)}$ & $( 0.092 \rlap{)}$ \\
&  &  &  &  &  &  &  &  &  &  &  \\
$y_{t-2}$ & $0.917 $ & $-0.379 $ & $0.678 $ &  & $1.332 $ & $-0.286 $ & $%
0.558 $ &  & $1.137 $ & $-0.319 $ & $0.648 $ \\
& $( 0.048 \rlap{)}$ & $( 0.072 \rlap{)}$ & $( 0.081 \rlap{)}$ &  & $( 0.053 %
\rlap{)}$ & $( 0.080 \rlap{)}$ & $( 0.066 \rlap{)}$ &  & $( 0.036 \rlap{)}$
& $( 0.054 \rlap{)}$ & $( 0.046 \rlap{)}$ \\
&  &  &  &  &  &  &  &  &  &  &  \\
Children & $-0.260 $ & $-0.410 $ & $-1.926 $ &  & $-0.143 $ & $0.112 $ & $%
-0.188 $ &  & $-0.223 $ & $-0.192 $ & $-1.209 $ \\
& $( 0.018 \rlap{)}$ & $( 0.069 \rlap{)}$ & $( 0.282 \rlap{)}$ &  & $( 0.024 %
\rlap{)}$ & $( 0.100 \rlap{)}$ & $( 0.251 \rlap{)}$ &  & $( 0.014 \rlap{)}$
& $( 0.051 \rlap{)}$ & $( 0.219 \rlap{)}$ \\
&  &  &  &  &  &  &  &  &  &  &  \\
Married & $0.136 $ & $0.022 $ & $-0.193 $ &  & $0.411 $ & $0.534 $ & $-0.019
$ &  & $0.393 $ & $0.269 $ & $-0.121 $ \\
& $( 0.095 \rlap{)}$ & $( 0.184 \rlap{)}$ & $( 0.028 \rlap{)}$ &  & $( 0.083 %
\rlap{)}$ & $( 0.203 \rlap{)}$ & $( 0.025 \rlap{)}$ &  & $( 0.061 \rlap{)}$
& $( 0.145 \rlap{)}$ & $( 0.022 \rlap{)}$ \\
&  &  &  &  &  &  &  &  &  &  &  \\
Sp.Inc & $-0.015 $ & $-0.050 $ & $0.255 $ &  & $0.023 $ & $-0.008 $ & $0.179
$ &  & $-0.022 $ & $-0.046 $ & $0.093 $ \\
& $( 0.007 \rlap{)}$ & $( 0.014 \rlap{)}$ & $( 0.210 \rlap{)}$ &  & $( 0.008 %
\rlap{)}$ & $( 0.019 \rlap{)}$ & $( 0.316 \rlap{)}$ &  & $( 0.005 \rlap{)}$
& $( 0.011 \rlap{)}$ & $( 0.183 \rlap{)}$ \\
&  &  &  &  &  &  &  &  &  &  &  \\
&  &  &  &  &  &  &  &  &  &  &
\end{tabular}
}}
\begin{tablenotes}
      \small
      \item The estimation also includes 11 time dummies. Standard error for the GMM and Logit Fixed Effects Estimators are calculated as the interquartile range of 1,000 bootstrap replications divided by 1.35.
\end{tablenotes}
\end{table}

To estimate the AR(2)\ version of the model, we apply the moment conditions provided
in Appendix~\ref{sec:moments_p2T4} to all consecutive sequences of six
outcomes (treating the first two as initial conditions). 
 The moment functions are scaled as described in the Monte Carlo simulations in Appendix~B.1 of \cite{Honore2022moment}.
The results are
presented in Table \ref{Table: Empirical AR2}. The most interesting finding
is that for all three samples, the GMM estimator of $\left(\gamma_1,\gamma_2%
\right)$ is between the maximum likelihood estimator that ignores the fixed
effects, and the one that estimates a fixed effect for each individual. This
suggests that unobserved individual-specific heterogeneity is important in
this example. Economically, it is also interesting that for each estimation
method, the estimates of $\left(\gamma_1,\gamma_2\right)$ are quite similar
across the three samples.

\section{Conclusion}
\label{sec:conc}

\cite{bonhomme2012functional} proposed a general approach for constructing moment restrictions in nonlinear panel data models that do not depend on individual-specific effects. In this paper, we have operationalized this in models with discrete outcomes by first presenting a blueprint for deciding whether  such moment conditions exist, and then an approach for actually finding analytic expressions for the moment conditions. 

We have used our approach to derive all the moment conditions for the panel
logit AR(1) model that are free of the fixed effects, and we have employed those moment conditions
to show identification of the common model parameters
and to obtain a GMM estimator that is useful and
performs well in practice. 
 The immediate practical relevance of this paper is therefore
for the dynamic panel logit model (both AR(1) and AR(2) models are estimated in an empirical application).

While part of this paper emphasises binary logit models, the  methods explained in Section~\ref{SEC: Incidental parameter free moment conditions} and \ref{sec:Derivation}
for exploring and deriving 
moment conditions are applicable for more general panel models,
as illustrated by the examples provided in Section~\ref{sec:Examples}.
Exploring such moment conditions in other interesting   models is a research agenda that has only started (e.g.\ \citealt{honore2021dynamic},  \citealt{davezies2022fixed}),
and a lot more future work should be done   to provide
useful new estimation methods in various discrete choice panel models.

\setstretch{1.3}
\setlength{\bibsep}{4pt} %
\ifx\undefined\BySame
\newcommand{\BySame}{\leavevmode\rule[.5ex]{3em}{.5pt}\ }
\fi
\ifx\undefined\textsc
\newcommand{\textsc}[1]{{\sc #1}}
\newcommand{\emph}[1]{{\em #1\/}}
\let\tmpsmall\small
\renewcommand{\small}{\tmpsmall\sc}
\fi

\appendix

\newpage
\setstretch{1.5}
\pagenumbering{roman}

\section{Online Appendix (not for publication)}

\subsection{Proofs of main text results}
\label{app:proofs}

Lemma~\ref{lemma:moments_p1T3} is a special case   of Theorem~\ref{th:AR1moments}, which is proven below.
Alternatively, Lemma~\ref{lemma:moments_p1T3} can be proved more directly by 
``brute-force calculations'', see Appendix~B.2.1 if \cite{Honore2022moment} of the supplementary appendix.

\begin{proof}[\bf Proof of Lemma~\protect\ref{lemma:INVERSION}]
The lemma holds trivially for $K=0$ when $s=\emptyset $ and $g_{\emptyset }:%
\mathbb{R}\rightarrow \mathbb{R}$ is a single increasing function, implying
that $g_{\emptyset }(\gamma )=0$ can at most have one solution. Consider $%
K\geq 1$ in the following. We follow a proof by contradiction. Assume that $%
(\beta _{1},\gamma _{1})\in \mathbb{R}^{K}\times \mathbb{R}$ and $(\beta
_{2},\gamma _{2})\in \mathbb{R}^{K}\times \mathbb{R}$ both solve $%
g_{s}(\beta _{1},\gamma _{1})=0$ and $g_{s}(\beta _{2},\gamma _{2})=0$, for
all $s\in \{-,+\}^{K}$, with $(\beta _{1},\gamma _{1})\neq (\beta
_{2},\gamma _{2})$. Our goal is to derive a contradiction between this and the
assumptions of the lemma. Without loss of generality, we assume that $\gamma
_{1}\leq \gamma _{2}$. Define $s^{\ast }\in \{-,+\}^{K}$ by
\begin{equation*}
s_{k}^{\ast }=\left\{
\begin{array}{ll}
+\; & \text{if}\;\beta _{1,k}\leq \beta _{2,k}, \\
-\; & \text{otherwise,}%
\end{array}%
\right.
\end{equation*}%
for all $k\in \{1,\ldots ,K\}$. By the monotonicity assumptions on $%
g_{s}(\beta ,\gamma )$ in the lemma, we have that $g_{s^{\ast }}(\beta
,\gamma )$ is strictly increasing in $\gamma $ and we have $\gamma _{1}\leq
\gamma _{2}$;  if $s_{k}^{\ast }=+$, then $g_{s^{\ast }}(\beta ,\gamma )$
is strictly increasing in $\beta _{k}$ and we have $\beta _{1,k}\leq \beta
_{2,k}$; and if $s_{k}^{\ast }=-$, then $g_{s^{\ast }}(\beta ,\gamma )$ is
strictly decreasing in $\beta _{k}$ and we have $\beta _{1,k}>\beta _{2,k}$.
Furthermore, one of these inequalities on the parameters must be strict,
because we have $(\beta _{1},\gamma _{1})\neq (\beta _{2},\gamma _{2})$. We
therefore conclude that
\begin{equation*}
g_{s^{\ast }}(\beta _{1},\gamma _{1})<g_{s^{\ast }}(\beta _{2},\gamma _{2}).
\end{equation*}%
This violates the assumption that $g_{s}(\beta _{1},\gamma _{1})=0$ and $%
g_{s}(\beta _{2},\gamma _{2})=0$. Thus, under the assumptions of the lemma
there cannot be two solutions of the system \eqref{systemEQ}.
\end{proof} 

\bigskip

\begin{proof}[\bf Proof of Theorem~\ref{th:id1}]
Consider $y_0=0$ and $q=0$.
Using the definition of the moment function $m^{(0)}(y,y_0,x,\beta ,\gamma )$ in Section~\ref{sec:FindSolutions}
and  the distribution of $Y|X,A$ implied by model \eqref{model} we find
\begin{align*}
   \frac{\partial \,  \overline{m}_{0,s}^{(0)}(\beta ,\gamma ) }
     {\partial \gamma}
   & =\mathbb{E}\left[ \left.
   \frac{\partial \, m^{(0)}(Y,0,X,\beta ,\gamma )} {\partial \gamma} \,\right|\,Y_{0}=0,\;X\in \mathcal{X}_{s}\right]
  \\
  &=  \mathbb{E}\left[ \left.
 \frac{\partial  \, \exp \left( \gamma +X_{21}^{\prime }\beta \right)  }{\partial \gamma}
 \, {\rm Pr}\left[ Y=(1,0,1) \, \big| \, Y_{0}=0,X,A\right]
   \,\right|\,Y_{0}=0,\;X\in \mathcal{X}_{s}\right]
  \\
  &=  \mathbb{E}\left[ \left.
   \exp \left( \gamma +X_{21}^{\prime }\beta \right)
 \,p((1,0,1),0,X,\beta_0 ,\gamma_0 ,A)
   \,\right|\,Y_{0}=0,\;X\in \mathcal{X}_{s}\right] > 0,
\end{align*}
where in the last step (to conclude that the expression is positive) we used that
$\,p((1,0,1),0,x,\beta_0 ,\gamma_0 ,\alpha) >0$
for all $x \in \mathbb{R}^{K \times 3}$ and $\alpha \in \mathbb{R}$.\footnote{%
Note that we have introduced the domain of $A$ to be $\mathbb{R}$.
 If we had introduced the domain of $A$ to be
$\mathbb{R} \cup \{\pm \infty\}$, then all other results in the paper hold completely unchanged,
but here we would need to impose the additional regularity condition that
 $A$ does not take values $\pm \infty$ with probability one, conditional on $Y_{0}=0$ and $X\in \mathcal{X}_{s}$,
 since otherwise we can have $p((1,0,1),0,X,\beta_0 ,\gamma_0 ,A)=0$.}
We have thus shown that
$ \overline{m}_{0,s}^{(0)}(\beta ,\gamma )$ is strictly increasing in $\gamma$.
Analogously, one can show that  $ \overline{m}_{0,s}^{(0)}(\beta ,\gamma )$ is strictly increasing in $\beta_k$
if $s_k=+$, and  strictly decreasing in $\beta_k$ if $s_k = -$, for all $k \in \{1,\ldots,K\}$, because of the result in  
\eqref{MonotonicityEmoments} above.

We can therefore apply Lemma~\ref{lemma:INVERSION} with $g_s(\beta,\gamma)$ equal to
$ \overline{m}_{0,s}^{(0)}(\beta ,\gamma )$ to find that the system of equations in \eqref{SystemTheoremID}
has at most one solution.  Using Lemma~\ref{lemma:moments_p1T3} we find that such a solution exists
and is given by $(\beta_0,\gamma_0)$.

For the other values of $q,y_{0}\in \{0,1\}$ 
we can analogously apply Lemma~\ref{lemma:INVERSION} with $g_s(\beta,\gamma)$ equal to
$\overline m^{(0)}_{1,s}(\beta,-\gamma)$,
$\overline m^{(1)}_{0,s}(-\beta,-\gamma)$, $\overline m^{(1)}_{1,s}(-\beta,\gamma)$, respectively.
\end{proof}

\subsubsection*{Key intermediate result for the proof of Lemma~\ref{lemma:Markov} and Theorem~\ref{th:AR1moments}}

Lemma~\ref{lemma:MASTER} below is a slight reformulation and  
generalization of Lemma~\ref{lemma:Markov} in the main text. Once we have established
Lemma~\ref{lemma:MASTER}, then the proofs of both Lemma~\ref{lemma:Markov} and Theorem~\ref{th:AR1moments} are straightforward. To present 
Lemma~\ref{lemma:MASTER}, it is useful to first introduce some additional notation:

Let  $  \widetilde Y_1,   \widetilde Y_2,   \widetilde Y_3 \in \{0,1\}   $,  $W_2 \in {\cal W}_2$, and $W_3 \in {\cal W}_3$ 
be random variables.  Let $\widetilde Y = (\widetilde Y_1,\widetilde Y_2,\widetilde Y_3)$,
and let $p(\widetilde y, w_2, w_3 ) \in [0,\infty)$ describe the joint distribution of $(\widetilde Y,W_2, W_3)$, that is,
 for all measurable subsets ${\cal Y}_* \subset \{0,1\}^3$, ${\cal W}^*_2 \subset {\cal W}_2$ and ${\cal W}^*_3 \subset {\cal W}_3$ we have
\begin{align*}
    {\rm Pr}\left( \widetilde Y  \in {\cal Y}_*  \; \; \& \; \;   W_2 \in {\cal W}^*_2  \; \; \& \; \;   W_3 \in {\cal W}^*_3 \right) 
    &=  \sum_{\widetilde y \in  {\cal Y}_* } \int_{w_2 \in {\cal W}^*_2} \int_{w_3 \in {\cal W}^*_3} \, p(\widetilde y, w_2, w_3 ) \, \mu(dw_2) \, \nu(dw_3) ,
\end{align*}
for appropriate probability measures $\mu$ and $\nu$ on $ {\cal W}_2$ and ${\cal W}_3$.
We assume that  we can decompose the joint distribution of $\widetilde Y$, $W_2$, $W_3$ as follows,
\begin{align}
  p(\widetilde y, w_2, w_3 ) &=  p_3(\widetilde y_3 \,| \,w_3) \; \,  g(w_3 \,|\, \widetilde y_2,w_2) \; \, p_2(\widetilde y_2 \,| \,w_2) \;\,  f(w_2 \,|\, \widetilde y_1)   \;\,  p_1(\widetilde y_1)   ,
     \label{AssMarkov}
\end{align}
where $\widetilde y = (\widetilde y_1,\widetilde y_2,\widetilde y_3)$, the functions
$p_3$, $g$, $p_2$, $f$ are appropriate transition probabilities/densities,
and $p_1(\widetilde y_1) =  {\rm Pr}\left( \widetilde Y_1 = \widetilde y_1  \right)$ is the marginal distribution of $ \widetilde Y_1$.
For $p_1(\widetilde y_1)$,  $p_2(\widetilde y_2 \,| \,w_2)$, $p_3(\widetilde y_3 \,| \,w_3) $ we impose:
\begin{align}
     p_1(\widetilde y_1  ) &=   \Lambda\left[ (2 \,\widetilde y_1-1) \, \pi_{1}    \right]  ,
\nonumber \\
    p_2(\widetilde y_2 \,|\, w_2) &=  \Lambda\left[ (2 \,\widetilde y_2-1) \, \pi_2(w_2)    \right] ,
\nonumber \\
     p_3(\widetilde y_3 \,| \, w_3) &=    \Lambda\left[ (2 \,\widetilde y_3-1) \, \pi_3(w_3)    \right] ,
   \label{AssLogit}  
\end{align}
where  $ \Lambda(\xi) :=  [1+\exp(-\xi)]^{-1}$ is the cumulative distribution function of the logistic distribution,
$\pi_1 \in \mathbb{R}$ is a constant,
and  $ \pi_{2} : \, {\cal W}_2 \rightarrow \mathbb{R}$ and  $ \pi_{3} : \, {\cal W}_3 \rightarrow \mathbb{R}$ are functions.\footnote{
Since  $\pi_1$, $\pi_2(w_2)$, $\pi_3(w_3)$ are unrestricted, the only substantial assumption that is actually imposed by \eqref{AssLogit} is that the conditional
probabilities should not be zero or one. 
}
The only assumption that we impose on $f(w_2 \,|\, \widetilde y_1) $ and $g(w_3 \,|\, \widetilde y_2,w_2)$ is that
\begin{align}
      g(w_3 \,|\, 1,w_2) = g(w_3 \,|\, 1) ,
      \label{RestrictionVW}
\end{align}
that is, conditional on $\widetilde Y_2=1$, the distribution of $W_3$ is independent of $W_2$.
Apart from that, we only require that $f(w_2 \,|\, \widetilde y_1) $ and $g(w_3 \,|\, \widetilde y_2,w_2)$ are conditional probability 
distributions, which sum to one:
\begin{align}
    \int_{w_2 \in {\cal W}_2} \,  f(w_2 \,|\, \widetilde y_1) \, \mu(dw_2)&=1,
    &
    \int_{w_3 \in {\cal W}_3} \,  g(w_3 \,|\, \widetilde y_2,w_2) \, \nu(dw_3)&=1.
    \label{NormalizeVW}
\end{align}    
Note that if we would strengthen \eqref{RestrictionVW} to $ g(w_3 \,|\, \widetilde y_2,w_2) = g(w_3 \,|\, \widetilde y_2)$,
for $\widetilde y_2 \in \{0,1\}$, then \eqref{AssMarkov} would be equivalent to the Markov chain
condition imposed in Lemma~\ref{lemma:Markov}
($      \widetilde Y_1
        \overset{f}{\longrightarrow}     W_2
        \overset{p_2}{\longrightarrow}  \widetilde  Y_2   
        \overset{g}{\longrightarrow}     W_3
        \overset{p_3}{\longrightarrow}  \widetilde  Y_3$).
However, since here we only impose \eqref{RestrictionVW}, we  also allow for dependence of $W_2$ and $W_3$, conditional on $  \widetilde  Y_2 = 0 $.
This generalization to a non-Markovian structure is crucial for the ordered 
logit model in \cite{honore2021dynamic}, but is less important for the binary
logit model discussed in this paper (See Appendix~\ref{app:MoreGeneralPredetermined} below for further discussion).
Finally, we define
$m : \{0,1\}^3 \times  {\cal W}_2 \times {\cal W}_3  \rightarrow \mathbb{R}$ 
by
 \begin{align}
  m(\widetilde y, w_2, w_3)   &:= 
    \left\{  \begin{array}{ll}
           \exp\left[   \pi_1  - \pi_2(w_2)  \right]     & \text{if }   \widetilde y =(0,1,0), 
            \\
         \exp\left[   \pi_1  -  \pi_3(w_3)  \right]     & \text{if }  \widetilde y =(0,1,1),        
            \\
         -1  & \text{if }  (\widetilde y_1, \widetilde y_2)=(1,0) ,                        
            \\
     \exp\left[    \pi_3(w_3)  - \pi_2(w_2) \right]  - 1  & \text{if }  \widetilde y =(1,1,0) , 
            \\
                      0 & \text{otherwise}.
       \end{array} \right.
     \label{DefMgeneral}  
\end{align}

\begin{lemma}
    \label{lemma:MASTER}
    Let $\pi_1 \in \mathbb{R}$, $ \pi_{2} : \, {\cal W}_2 \rightarrow \mathbb{R}$ and  $ \pi_{3} : \, {\cal W}_3 \rightarrow \mathbb{R}$.
     Let the random variables $\widetilde Y \in \{0,1\}^3$, $W_2 \in {\cal W}_2$, $W_3 \in {\cal W}_3$
      be such that their distributions satisfy  \eqref{AssMarkov}, \eqref{AssLogit}, \eqref{RestrictionVW}, \eqref{NormalizeVW},
      and let $m : \{0,1\}^3 \times  {\cal W}_2 \times {\cal W}_3  \rightarrow \mathbb{R}$   be defined by \eqref{DefMgeneral}.      
       Then  we have
     \begin{align*}
          \mathbb{E}\left[ m(\widetilde Y, W_2, W_3)   \right] = 0.
     \end{align*}
\end{lemma}

\begin{proof}
    Define
          \begin{align*}
           h(\widetilde y_1,\widetilde y_2,w_2, w_3)
         &:=  
            \,  \sum_{\widetilde y_3 \in \{0,1\}} \, m(\widetilde y, w_2, w_3)  \; \, p_3(\widetilde y_3 \,| \,w_3) \; \,    p_2(\widetilde y_2 \,| \,w_2) 
            \; \,   p_1(\widetilde y_1)    ,
       \end{align*}
    where $\widetilde y = ( \widetilde y_1, \widetilde y_2, \widetilde y_3)$.
     By using the expressions for the functions $p_1$, $p_2$, $p_3$,  and $m$ in \eqref{AssLogit} and \eqref{DefMgeneral}
     we find that       
   \begin{align*}
           h(\widetilde y_1,\widetilde y_2,w_2, w_3)
          &=    \left\{  \begin{array}{ll}
           \Lambda(\pi_1)  \, \Lambda[-\pi_3(w_3)]  & \text{if }  (\widetilde y_1, \widetilde y_2)=(0,1), 
            \\
         - p_2(\widetilde y_2 \,| \,w_2)  \, p_1(\widetilde y_1)   & \text{if }  (\widetilde y_1, \widetilde y_2)=(1,0) ,                        
            \\
     \Lambda(\pi_1)  \,  \Lambda[\pi_3(w_3)]  \,  - p_2(\widetilde y_2 \,| \,w_2)  \, p_1(\widetilde y_1)    & \text{if }  (\widetilde y_1, \widetilde y_2)=(1,1) , 
            \\
                      0 & \text{otherwise},
       \end{array} \right.    
      \\
        &= 
        \underbrace{
         \mathbbm{1}\left\{ \widetilde y_2 = 1 \right\}
        \,   \Lambda(\pi_1)  \, \Lambda[(2\widetilde y_1-1) \pi_3(w_3)]  
        }_{\displaystyle =:  h_1(\widetilde y_1,\widetilde y_2,w_3)}
       \, \underbrace{  - \,   \mathbbm{1}\left\{ \widetilde y_1 = 1 \right\} \, p_2(\widetilde y_2 \,| \,w_2)  \,  p_1(\widetilde y_1) 
         }_{\displaystyle =:  h_2(\widetilde y_1,\widetilde y_2,w_2)} ,
     \end{align*}
    where we have decomposed $h(\widetilde y_1,\widetilde y_2,w_2, w_3)$ into the sum of $h_1(\widetilde y_1,\widetilde y_2,w_3)$,
     which does not depend on $w_2$, and of $ h_2(\widetilde y_1,\widetilde y_2,w_2)$, which does not depend on $w_3$.
     Note that the term $\Lambda[(2\widetilde y_1-1) \pi_3(w_3)] $ in $h_1(\widetilde y_1,\widetilde y_2,w_3)$
     is identical to $p_3(\widetilde y_3 \,| \,w_3)$, but with $\widetilde y_3$ replaced by $\widetilde y_1$.
     Also using  \eqref{RestrictionVW} and \eqref{NormalizeVW} we find
     \begin{align*}
     & \sum_{\widetilde y_1 \in \{0,1\}} 
       \sum_{\widetilde y_2 \in \{0,1\}} 
         \int_{w_2 \in {\cal W}_2} \, \int_{w_3 \in {\cal W}_3}  \,
          h_1(\widetilde y_1,\widetilde y_2,w_3) \; \,  g(w_3 \,|\, \widetilde y_2,w_2) \; \,  f(w_2 \,|\, \widetilde y_1)   
          \, \mu(dw_2)  \, \nu(dw_3)
\\
       &=  \Lambda(\pi_1)
          \sum_{\widetilde y_1 \in \{0,1\}} 
         \int_{w_2 \in {\cal W}_2} \, \int_{w_3 \in {\cal W}_3}  \,
        \, \Lambda[(2\widetilde y_1-1) \pi_3(w_3)]  \; \,  g(w_3 \,|\,1) \; \,  f(w_2 \,|\, \widetilde y_1)   \;\,   
          \, \mu(dw_2)  \, \nu(dw_3)
\\
       &=  \Lambda(\pi_1)
     \sum_{\widetilde y_1 \in \{0,1\}}        \int_{w_3 \in {\cal W}_3}
        \, \Lambda[(2\widetilde y_1-1) \pi_3(w_3)]  \; \,  g(w_3 \,|\,1) \,  
       \underbrace{   \int_{w_2 \in {\cal W}_2}     \, f(w_2 \,|\, \widetilde y_1)      \, \mu(dw_2)}_{=1}  \, \nu(dw_3)
\\
       &=  \Lambda(\pi_1)
         \int_{w_3 \in {\cal W}_3} 
        \underbrace{  \sum_{\widetilde y_1 \in \{0,1\}}     \, \Lambda[(2\widetilde y_1-1) \pi_3(w_3)] }_{=1}  \; \,  g(w_3 \,|\,1)  \,   \nu(dw_3)
\\
       &=  \Lambda(\pi_1)
       \underbrace{  \int_{w_3 \in {\cal W}_3}    \,  g(w_3 \,|\,1)  \,   \nu(dw_3) }_{=1}
\\
       &=  \Lambda(\pi_1) .
     \end{align*}
     Similarly, we calculate
 \begin{align*}
     & \sum_{\widetilde y_1 \in \{0,1\}} 
       \sum_{\widetilde y_2 \in \{0,1\}} 
         \int_{w_2 \in {\cal W}_2} \, \int_{w_3 \in {\cal W}_3}  \,
          h_2(\widetilde y_1,\widetilde y_2,w_2) \; \,  g(w_3 \,|\, \widetilde y_2,w_2) \; \,  f(w_2 \,|\, \widetilde y_1)  
          \, \mu(dw_2)  \, \nu(dw_3)
\\
   &=     \sum_{\widetilde y_1 \in \{0,1\}} 
      \sum_{\widetilde y_2 \in \{0,1\}} 
         \int_{w_2 \in {\cal W}_2}   \,
          h_2(\widetilde y_1,\widetilde y_2,w_2) \underbrace{
         \int_{w_3 \in {\cal W}_3} \,  g(w_3 \,|\, \widetilde y_2,w_2) \, \nu(dw_3)      
          }_{=1}  f(w_2 \,|\, \widetilde y_1)   
          \, \mu(dw_2)  
\\
   &=   -  \, p_1(1)
     \int_{w_2 \in {\cal W}_2}   \,     \underbrace{ \sum_{\widetilde y_2 \in \{0,1\}}   p_2(\widetilde y_2 \,| \,w_2)  }_{=1}
        \;\,     f(w_2 \,|\,1)  
          \, \mu(dw_2)  
\\
   &=   -  \, p_1(1)
    \underbrace{ \int_{w_2 \in {\cal W}_2}   \,         f(w_2 \,|\,1)      \, \mu(dw_2)  }_{=1}
          \\
          &= -  \, \Lambda(\pi_1) .
     \end{align*}
    Combining the results in the last two displays gives
     \begin{align*}
     & \sum_{\widetilde y_1 \in \{0,1\}} 
       \sum_{\widetilde y_2 \in \{0,1\}} 
         \int_{w_2 \in {\cal W}_2} \, \int_{w_3 \in {\cal W}_3}  \,
          h(\widetilde y_1,\widetilde y_2,w_2, w_3) \; \,  g(w_3 \,|\, \widetilde y_2,w_2) \; \,  f(w_2 \,|\, \widetilde y_1)   
          \, \mu(dw_2)  \, \nu(dw_3) = 0,
     \end{align*}
     and by the definition of $ h(\widetilde y_1,\widetilde y_2,w_2, w_3)$ this is equivalent to 
     $  \mathbb{E}\left[ m(\widetilde Y, W_2, W_3)   \right] =0$, which is what we wanted to show.   
\end{proof}

\bigskip
We are now ready to prove the remaining main text results.
\bigskip

\begin{proof}[\bf Proof of Lemma~\ref{lemma:Markov}]
    Let the assumptions of Lemma~\ref{lemma:Markov} hold.
   We are going to verify that this implies that all
   the assumptions of Lemma~\ref{lemma:MASTER} hold,
    conditional on $(W_1,X,A)$. 
    We condition on $(W_1,X,A)=(w_1,x,a)$ in all the following
    stochastic statements, and when verifying 
    the assumptions of Lemma~\ref{lemma:MASTER}, 
    we write    
    $p_t(\widetilde y_t \, |\, w_t,x,\alpha)$, $t \in \{1,2,3\}$,
    instead of 
    $p_1(\widetilde y_1)$,
    $p_2(\widetilde y_2 \, | \, w_2)$,
    $p_3(\widetilde y_3 \, | \, w_3)$.
    The additional arguments $w_1$, $x$, $a$ do not matter here, since they are held constant (i.e.\ conditioned on). 
    By the same argument, we write 
     $\pi_t(w_t,x,\alpha)$, $t \in \{1,2,3\}$, here
    instead of 
    $\pi_1$,
    $\pi_2( w_2)$,
    $\pi_3( w_3)$.
    The Markov chain assumption in Lemma~\ref{lemma:Markov} 
     implies that conditions \eqref{AssMarkov} 
     and \eqref{RestrictionVW}
     hold.\footnote{
In fact, we have $ g(w_3 \,|\, \widetilde y_2,w_2) = g(w_3 \,|\, \widetilde y_2)$ not only for  $\widetilde y_2=1$
but also for $\widetilde y_2=0$ here, which is stronger than
required for applying Lemma~\ref{lemma:MASTER}.
      }
   For $t \in \{1,2,3\}$ we  define
    \begin{align}
     \pi_t(w_t,x,\alpha) :=
\Lambda^{-1}\Big[ p_t \left(1 \, \big| \,w_t,x,\alpha \right)  \Big],
       \label{DefPiTproof}
    \end{align}
    which guarantees that    
   \eqref{AssLogit} holds.
   Finally, \eqref{NormalizeVW} is satisfied for
   any conditional probability.\footnote{
   In other words, explicitly imposing \eqref{NormalizeVW} is redundant in Lemma~\ref{lemma:MASTER} since it is always
   satisfied in the setup described. We list this
   condition explicitly in the statement of Lemma~\ref{lemma:MASTER} since
   it is used in the proof. For example,
   non-negativity of 
    $f(w_2 \,|\, \widetilde y_1) $ and $g(w_3 \,|\, \widetilde y_2,w_2)$ is also implied by the setup described, but is not
    actually 
    required for the proof.
   }
   We can thus apply Lemma~\ref{lemma:MASTER} 
   to find that 
   \begin{align*}
  \mathbb{E}\left[ m^{(1)}(w_1,\widetilde Y_1,W_2,\widetilde Y_2,W_3,\widetilde Y_3,x,a) \, \Big| \, W_1=w_1, \, X=x, \, A=a  \right] = 0,
  \end{align*}
where 
$m^{(1)}(w_1,\widetilde y_1,w_2,\widetilde y_2,w_3,\widetilde Y_3,x,a)$
is equal to $m(\widetilde y, w_2, w_3)$
defined in \eqref{DefMgeneral}, with $\widetilde y=(\widetilde y_1,\widetilde y_2,\widetilde y_3)$,
and with the additional arguments
$w_1$, $x$, $a$ added through
$\pi_t(w_t,x,\alpha)$, as explained above, that is,
 \begin{align*}
 & m^{(1)}(w_1,\widetilde y_1,w_2,\widetilde y_2,w_3,\widetilde y_3,x,\alpha)   
 \\
 & \quad = \left\{  \begin{array}{ll}
           \exp\left[   \pi_1(w_1,x,\alpha)  - \pi_2(w_2,x,\alpha)  \right]     & \text{if }   \widetilde y =(0,1,0), 
            \\
         \exp\left[   \pi_1(w_1,x,\alpha)  -  \pi_3(w_3,x,\alpha)  \right]     & \text{if }  \widetilde y =(0,1,1),        
            \\
         -1  & \text{if }  (\widetilde y_1, \widetilde y_2)=(1,0) ,                        
            \\
     \exp\left[    \pi_3(w_3,x,\alpha)  - \pi_2(w_2,x,\alpha) \right]  - 1  & \text{if }  \widetilde y =(1,1,0) , 
            \\
                      0 & \text{otherwise},
       \end{array} \right.
 \\
 & \quad =   -   \mathbbm{1}\left\{ \widetilde y_1 = 1 \right\}
  + \left\{  \begin{array}{ll}
           \exp\left[   \pi_1(w_1,x,\alpha)  - \pi_2(w_2,x,\alpha)  \right]     & \text{if }   \widetilde y =(0,1,0), 
            \\
         \exp\left[   \pi_1(w_1,x,\alpha)  -  \pi_3(w_3,x,\alpha)  \right]     & \text{if }  \widetilde y =(0,1,1),        
            \\
     \exp\left[    \pi_3(w_3,x,\alpha)  - \pi_2(w_2,x,\alpha) \right]  
  & \text{if }  \widetilde y =(1,1,0) , 
            \\
    1  
  & \text{if }  \widetilde y =(1,1,1) , 
            \\
                      0 & \text{otherwise},
       \end{array} \right.
 \\
& \quad =       -   \mathbbm{1}\left\{ \widetilde y_1 = 1 \right\} 
   \\ &     \qquad     
   +
        \mathbbm{1}\left\{ \widetilde y_2 = 1 \right\} \,        
      \exp  \left\{  
      \sum_{t=1}^3 \widetilde y_{t} \bigg[
      \pi_{\delta(t)}\left(w_{\delta(t)},x,\alpha \right)
     -  \pi_t(w_t,x,\alpha)  
             \bigg]
         \right\}
 \\
& \quad =       -   \mathbbm{1}\left\{ \widetilde y_1 = 1 \right\} 
   \\ &     \qquad     
   +
        \mathbbm{1}\left\{ \widetilde y_2 = 1 \right\} \,        
      \exp  \left\{ \frac 1 2 
      \sum_{t=1}^3 (2\widetilde y_t-1) \bigg[
         \pi_{\delta(t)}\left(w_{\delta(t)},x,\alpha \right)
     -
           \pi_t(w_t,x,\alpha)  
             \bigg]
         \right\}
 \\
& \quad =   -   \mathbbm{1}\left\{ \widetilde y_1 = 1 \right\} 
   \\ &     \qquad     
   +
        \mathbbm{1}\left\{ \widetilde y_2 = 1 \right\} \,        
      \exp  \left( \frac 1 2 
      \sum_{t=1}^3 \bigg\{
     \Lambda^{-1}\Big[ 
     p_{\delta(t)} \left(\widetilde y_t \, \big| \,w_{\delta(t)},x,\alpha \right)   \Big]
     -
             \Lambda^{-1}\Big[ p_t \left(\widetilde y_{t} \, \big| \,w_t,x,\alpha \right)  \Big]
             \bigg\}
         \right) ,
\end{align*}
where in the last step we used that
\eqref{DefPiTproof} 
together with
$\sum_{y \in \{0,1\}} p_{t} \left( y   \big|  w_t,x,\alpha \right)  =1 $
implies that 
$ p_{t} \left( y   \big|  w_t,x,\alpha \right) 
= \Lambda [ (2y-1)  
         \, \pi_{}\left(w_{t},x,\alpha \right) ]$,
         for all $y \in \{0,1\}$ and $t\in \{1,2,3\}$.
We have therefore shown the result of 
Lemma~\ref{lemma:Markov} for $q=1$.
The result for $q=0$ directly follows from this by applying the 
symmetry transformation $\widetilde y_t \, \leftrightarrow \, 1-\widetilde y_t$, and
$\pi_t(w_t,x,\alpha) \, \leftrightarrow \, - \pi_t(w_t,x,\alpha)$.
\end{proof}

\bigskip

\begin{proof}[\bf Proof of Theorem~\ref{th:AR1moments}]
Let the assumptions of Theorem~\ref{th:AR1moments} be satisfied.
By choosing 
$(\widetilde Y_1, \widetilde  Y_2, \widetilde  Y_3)
=( Y_t,  Y_s,  Y_r)$
and 
$(W_1, W_2, W_3)
=( Y_{t-1},  Y_{s-1},  Y_{r-1})$, it is then easy to verify that the
conditions of Lemma~\ref{lemma:Markov} are satisfied,
when conditioning on $(X,A,Y_{0},\, Y_{1},  \ldots \allowbreak,\allowbreak Y_{t-1})$.
The only substantial 
assumption to verify here is the Markov chain condition,
which immediately follows from the AR(1) structure of $Y_t$.
The panel AR(1) model restriction in \eqref{model} (after dropping the index $i$)
then implies that 
the moment function
$ m^{(q)}(w_1,\widetilde y_1,w_2,\widetilde y_2,w_3,\widetilde y_3,x,\alpha) $
in the lemma now reads
\begin{align*}
   &  m^{(q)(t,s,r)}   (y,y_0,x,\beta_0 ,\gamma_0 )
  \\   &  \quad  =    -   \mathbbm{1}\left\{ y_t = q \right\} 
 +
        \mathbbm{1}\left\{ y_s = q \right\} \,        
      \exp  \left\{ \frac 1 2 
    \left[  (2y_t-1) z_{rt} 
          + (2y_s-1) z_{ts} 
          + (2y_r-1) z_{sr} 
          \right] \right\}
  \\   &  \quad  =    -   \mathbbm{1}\left\{ y_t = q \right\} 
 +
        \mathbbm{1}\left\{ y_s = q \right\} \,        
      \exp   
    \left(   y_t \, z_{rt} 
          +  q \, z_{ts} 
          + y_r \, z_{sr} 
          \right) , 
\end{align*}
where $z_{ts} = z_{ts}(y,y_{0},x,\beta_0 ,\gamma_0) 
= (x_t-x_s)' \beta_0 + (y_{t-1} - y_{s-1}) \gamma_0$
was already defined in the main text.
It is easy to verify that the moment functions in the last display
coincide with those
defined in \eqref{MomentsGeneralT} of the main text.
Lemma~\ref{lemma:Markov} 
therefore guarantees that the conditional moment condition
in Theorem~\ref{th:AR1moments} holds.
\end{proof}

\subsection{Remarks on more general predetermined regressors} 
\label{app:MoreGeneralPredetermined}

Everywhere in the main text of the paper we have assumed
that $X_{it}$ is strictly exogenous (because the model
for the outcomes was specified conditional on $X$).
However, our results for the panel AR(1) model can still be applied to
certain types of predetermined regressors.
In particular, arbitrary feedback from
last period outcomes $Y_{i,t-1}$ into $X_{it}$ 
can be allowed for, as long as $X_{it}$ otherwise only
depends on strictly exogenous variables. That is,
for an arbitrary function $h(\cdot,\cdot)$
and an unobserved strictly exogenous variable $\widetilde X_{it}$ 
we consider
\begin{align}
   X_{it} = h(Y_{i,t-1},\widetilde X_{it}). 
   \label{GeneralizedX}
\end{align}
Note that \eqref{modelGENERALIZED} 
is a special case of this where 
$\widetilde X_{it} = X_{it}$ is observed 
and the regressors that actually enter into the model
are given by the functions $(1-Y_{i,t-1}) \, X_{it}$ and
$Y_{i,t-1} \, X_{it}$. The novel point here
is that $\widetilde X_{it}$ can be unobserved.
In this setting, the model in \eqref{model}
needs to be changed to
\begin{equation*}
\mathrm{Pr}\left( Y_{it}=1\,\big|\,Y_{i}^{t-1},\widetilde X_{i},A_{i}\right) =\frac{%
\exp \big(X_{it}^{\prime }\,\beta +Y_{i,t-1}\,\gamma +A_{i}\big)}{1+\exp %
\big(X_{it}^{\prime }\,\beta +Y_{i,t-1}\,\gamma +A_{i}\big)},
\end{equation*}
which by the law of iterated expectations implies 
\begin{equation*}
\mathrm{Pr}\left( Y_{it}=1\,\big|\,Y_{i}^{t-1},X^{t}_{i},A_{i}\right) =\frac{%
\exp \big(X_{it}^{\prime }\,\beta +Y_{i,t-1}\,\gamma +A_{i}\big)}{1+\exp %
\big(X_{it}^{\prime }\,\beta +Y_{i,t-1}\,\gamma +A_{i}\big)},
\end{equation*}
where $X^{t}_{i}=(X_{it},X_{i,t-1},X_{i,t-2},\ldots)$.
The last two displays formalize what we mean by $\widetilde X_{it}$ being
strictly exogenous and $X_{it}$ being predetermined.

The moment functions in \eqref{MomentsGeneralT} can be shown to be valid
in this setting
by
applying Lemma~\ref{lemma:Markov} with $W_t = (Y_{t-1},X_t)$, and with the conditioning
variables $X$ replaced by $\widetilde X=(\widetilde X_1,\ldots, \widetilde X_T)$. 
For covariates of the form \eqref{GeneralizedX},
the conclusion of Theorem~\ref{th:AR1moments}
then gets modified to
\begin{align*}
\mathbb{E}\left[  m^{(q)(t,s,r)}(Y,Y_0,X,\beta
_{0},\gamma _{0})\,\big|\,Y_{0},\, Y_{1},\ldots ,Y_{t-1}, \,
X_1,\, X_{2},\ldots ,X_t, \,
\widetilde X,\,A  \right] & =0,
\end{align*}
which by the law of iterated expectations implies
\begin{align}
\mathbb{E}\left[  m^{(q)(t,s,r)}(Y,Y_0,X,\beta
_{0},\gamma _{0})\,\big|\,Y_{0},\, Y_{1},\ldots ,Y_{t-1}, \,
X_1,\, X_{2},\ldots ,X_t \right] & =0.
  \label{MomentPredetermined}
\end{align}
Here, as in Theorem~\ref{th:AR1moments}, we consider $t<s<r$,
but the difference is that  we only condition on covariates up
to time period $t$ in this moment condition, 
while in the main text we always conditioned on the whole $X$.
The identification arguments and the GMM estimator
in the main text would have to be adjusted
accordingly, but we do not explore this generalization
further here.

\bigskip

Another possible extension of our results to predetermined covariates is as follows: Consider the model 
in \eqref{modelGENERALIZED}, but set $\beta_1=0$ and make $X_{it}$
predetermined, that is,
\begin{equation}
\mathrm{Pr}\left( Y_{it}=1\,\big|\,Y_{i}^{t-1},X^t_{i},A_{i}\right) =\frac{%
\exp \big[
(1-Y_{i,t-1}) X_{it}^{\prime }\,\beta_0
  +Y_{i,t-1}\,\gamma +A_{i}\big)}{1+\exp %
\big((1-Y_{i,t-1}) X_{it}^{\prime }\,\beta_0
  +Y_{i,t-1}\,\gamma +A_{i}\big]}.  \label{modelGENERALIZEDpredetermined}
\end{equation}%
It turns out that the moment conditions
in \eqref{MomentPredetermined} remain valid in this model 
for $q=1$  
even for general
pre-determined regressors $X_{it}$.
The reason for this is that
Lemma~\ref{lemma:MASTER} in the appendix
is more general than  
Lemma~\ref{lemma:Markov} in the main text.
Specifically, in Lemma~\ref{lemma:MASTER}, the condition
\eqref{RestrictionVW} only demands
$ g(w_3 \,|\, 1,w_2) = g(w_3 \,|\, 1)$,
which rules out direct feedback from $W_2$ into $W_3$ whenever $\widetilde Y_2 = 1$,
but this still allows for arbitrary feedback
from $W_2$  into $W_3$  whenever $\widetilde Y_2  = 0$.
By following the proof of  Theorem~\ref{th:AR1moments}
(where $(\widetilde Y_1, \widetilde  Y_2, \widetilde  Y_3)
=( Y_t,  Y_s,  Y_r)$),  but setting
$(W_1, W_2, W_3)
=( (Y_{t-1},(1-Y_{t-1}) X_t),  (Y_{s-1},(1-Y_{s-1}) X_s),  (Y_{r-1},(1-Y_{r-1}) X_r))$,
and conditioning on $W_1$,
we find that for
  general pre-determined $X_{it}$
  the moment condition
in \eqref{MomentPredetermined} still holds for $q=1$,
as long as  $X_{t}$   only enters into the model
for $Y_{t}$ through $(1-Y_{t-1}) X_{t}$, as in \eqref{modelGENERALIZEDpredetermined}.

Of course, model \eqref{modelGENERALIZEDpredetermined} and pre-determined regressors of the form
\eqref{GeneralizedX} are both quite restrictive.
This is why we only briefly discuss those possible extensions
to predetermined regressors here in the appendix. Nevertheless,
the possibility of using functional differencing ideas
to make progress on non-linear panel model with 
more general predetermined regressors is quite exciting, see also
\cite{Sequential2022}.

\subsection{Fixed effect logit AR($p$) models with $p>1$}
\label{sec:ARp}

In this appendix section, we consider logit AR($p$) models, that is, we
generalize the model in \eqref{model} to
\begin{equation}
\mathrm{Pr}\left( Y_{it}=1\,\big|\,Y_{i}^{t-1},X_{i},A_{i},\beta ,\gamma
\right) =\frac{\exp \big(X_{it}^{\prime }\,\beta +\sum_{\ell
=1}^{p}\,Y_{i,t-\ell }\,\gamma _{\ell }+A_{i}\big)}{1+\exp \big(%
X_{it}^{\prime }\,\beta +\sum_{\ell =1}^{p}\,Y_{i,t-\ell }\,\gamma _{\ell
}+A_{i}\big)}, %
\end{equation}%
where $\gamma =(\gamma _{1},\ldots ,\gamma _{p})^{\prime }$. We assume that
the autoregressive order $p\in \{2,3,4,\ldots \}$ is known, and that
outcomes $Y_{it}$ are observed for time periods $t=t_{0},\ldots ,T$, with $%
t_{0}=1-p$. Thus, the total number of time periods for which outcomes are
observed is $T_{\mathrm{obs}}=T+p$, consisting of $T$ periods for which the model applies
and $p$ periods to observe the initial conditions. We maintain the
definition $Y_{i}=(Y_{i1},\ldots ,Y_{iT})$, but the initial conditions are
now described by the vector $Y_{i}^{(0)}=(Y_{i,t_{0}},\ldots ,Y_{i0})$.
Analogous to~\eqref{DefProb}, we define
\begin{equation}
p_{y_{i}^{(0)}}(y_{i},x_{i},\beta ,\gamma ,\alpha _{i})=\prod_{t=1}^{T}\frac{%
\exp \big(x_{it}^{\prime }\,\beta +\sum_{\ell =1}^{p}\,y_{i,t-\ell }\,\gamma
_{\ell }+\alpha _{i}\big)}{1+\exp \big(x_{it}^{\prime }\,\beta +\sum_{\ell
=1}^{p}\,y_{i,t-\ell }\,\gamma _{\ell }+\alpha _{i}\big)},
    \label{DefProbGeneralOrder}
\end{equation}%
and we continue to denote the true model parameters by $\beta _{0}$ and $%
\gamma _{0}$. We again drop
the cross-sectional indices $i$ unless they are required.

Our main focus is again on finding moment conditions that are applicable to all
values of the parameters and all realizations of the regressors.
For a given autoregressive order $p$, one requires $T \geq 2+p$ (i.e.\ $T_{\rm obs} \geq 2+2p$) time periods to find
such general moment conditions, and the number of linearly independent
moment conditions available for each initial condition $y^{(0)}$
is then equal to
 \begin{align}
     \ell = 2^{T} - (T+1-p) \, 2^p .
     \label{MomentCountGeneral}
 \end{align}
Analogous to Section~\ref{SEC: On the number of moment conditions} one can use the
polynomial structure of \eqref{DefProbGeneralOrder} in $\exp(\alpha_i)$ to show that
at least that many moment conditions have to exist.  That \eqref{MomentCountGeneral} 
actually holds with equality can then be verified numerically.
 
In addition to those general moment conditions, 
which exist for all possible values of $\beta$, $\gamma$ and $x_{it}$,
there are additional ones
that only become available for special values of the parameters and of the regressors. For example,
for $T=4$ the AR(2) model  has $\ell=4$ general moments for each initial condition, but if $\gamma_2=0$ then the model becomes an AR(1) model with  $\ell=8$ available moments for each initial condition.\footnote{
Setting $\gamma_k = 0$ gives additional moments not only for $k=p$. For example, for the AR(2) model with $T=4$
one finds nine valid moment conditions for each initial condition when $\gamma_1=0$ and $\gamma_2 \neq 0$.
}
Furthermore, there are additional moment conditions available for special realizations of the regressors, and  those can provide identifying information for the parameters $\beta$ and $\gamma$ even when $T < 2+p$,\footnote{
Interestingly, this is not the case for AR(1) models, where for $\gamma \neq 0$ we have always found the same number
of available moment conditions, completely independent of the regressor value $x \in \mathbb{R}^{K \times T}$.
}
see Appendix~B.3.1 of \cite{Honore2022moment}.

\subsubsection{AR(2) models with $T=4$}
\label{sec:moments_p2T4}

Consider model \eqref{modelARp} with $p=2$ and $T=4$.
Again write $z_{t}(y,y_{0},x,\beta,\gamma )=x_{t}^{\prime }\,\beta +y_{t-1}\,\gamma_1 +y_{t-2}\,\gamma_2 $
 for the single index that describes how the parameters $\beta$ and $\gamma$ enter into the model
 at time period $t$,
 and  let $z_{ts}(y,y_{0},x,\beta ,\gamma)=z_{t}(y,y_{0},x,\beta ,\gamma )-z_{s}(y,y_{0},x,\beta ,\gamma )$
be its time-differences.
To save space, we drop the arguments of the differences in the following formulas and write $z_{ts}$ instead of $z_{ts}(y,y_{0},x,\beta ,\gamma)$.
One valid moment function for any
initial condition $y^{(0)} \in \{0,1\}^2$ and general covariate values $x\in \mathbb{R}^{K\times T}$ is then given by
\begin{align*}
m_{y^{(0)}}^{(a,2,4)}(y,x,\beta ,\gamma )& =\left\{
\begin{array}{@{\,}l@{\;\;}l}
\exp (z_{23} )-\exp (z_{43}  ) &
\text{if }y=(0,0,1,0), \\
\exp (z_{24}  )-1 & \text{if }y=(0,0,1,1), \\
-1 & \text{if }(y_{1},y_{2})=(0,1), \\
\exp (z_{41} +\gamma _{1}) & \text{if }%
(y_{1},y_{2},y_{3})=(1,0,0), \\
\exp (z_{41}  )\big[1+\exp (z_{23} )  -\exp (z_{43} )\big] & \text{if }y=(1,0,1,0), \\
\exp (z_{21} ) & \text{if }y=(1,0,1,1), \\
0 & \text{otherwise}.%
\end{array}%
\right.
\end{align*}
There are three additional moment functions $m_{y^{(0)}}^{(b,2,4)}(y,x,\beta ,\gamma )$, $m_{y^{(0)}}^{(c,2,4)}(y,x,\beta ,\gamma )$
and $m_{y^{(0)}}^{(d,2,4)}(y,x,\beta ,\gamma )$, and they are provided in 
Appendix~B.3.2 of \cite{Honore2022moment}.

\begin{lemma}
\label{lemma:moments_p2T4} If the outcomes $Y=(Y_{1},Y_{2},Y_{3},Y_{4})$ are
generated from model \eqref{modelARp} with $p=2$, $T=4$ and true parameters $%
\beta _{0}$ and $\gamma _{0}$, then we have for all $y^{(0)} \in \{0,1\}^2$, $x\in \mathbb{R}%
^{K\times 4}$, $\alpha \in \mathbb{R}$, and $\xi \in \{a,b,c,d\}$ that
\begin{equation*}
\mathbb{E}\left[ m_{y^{(0)}}^{(\xi ,2,4)}(Y,X,\beta _{0},\gamma _{0})\,\Big|%
\,(Y_{-1},Y_{0})=y^{(0)},\,X=x,\,A=\alpha \right] =0.
\end{equation*}
\end{lemma}

The proof of the lemma is discussed in Appendix~B.3.3 of \cite{Honore2022moment}.
Analogous to our discussion for the AR(1) model, one can apply GMM to those moment conditions.
Under suitable regularity conditions, this allows to estimate the parameters $\beta$ and $\gamma$ of the panel AR(2) logit
model at root-n-rate.

\subsubsection{AR(3) models with $T=5$}
\label{sec:moments_p3T5}

To illustrate the applicability of this general approach to constructing moment conditions further, we next consider a panel logit AR(3) model with $T=5$.
Since the model contains three lags, one needs to observe three
time periods as
initial conditions, which gives a total number of observations
required  of $T_{\rm obs} = 8$.
Let $z_{t}(y,y_{0},x,\beta,\gamma
)=x_{t}^{\prime }\,\beta +y_{t-1}\,\gamma_1 +y_{t-2}\,\gamma_2
+y_{t-3}\,\gamma_3$, and let $z_{ts}(y,y_{0},x,\beta
,\gamma)=z_{t}(y,y_{0},x,\beta ,\gamma )-z_{s}(y,y_{0},x,\beta ,\gamma )$.
We again drop the arguments of $z_{ts}$  in the following. There are then eight valid moment
functions for any initial condition $y^{(0)} \in \{0,1\}^3$ and general
covariate values $x\in \mathbb{R}^{K\times 5}$. The first of those is given by
\begin{equation*}
m_{y^{(0)}}^{(a,3,5)}(y,x,\beta ,\gamma )
\hspace{-2pt}
=
\hspace{-2pt}
\left\{
\begin{array}{@{}l@{\;}l}
\exp \left( z_{23}\right) -\exp \left( z_{53}\right)  & \text{if }%
y=(0,0,1,0,0), \\
\left[ \exp \left( z_{43}\right) -\exp \left( z_{45}\right) +1\right]  &  \\
\qquad \qquad \qquad \qquad  \times \left[ \exp \left( z_{25}\right) - 1 \right]
& \text{if }y=(0,0,1,0,1), \\
\exp \left( \gamma _{1}+z_{25}\right) -1 & \text{if }%
(y_{1},\ldots,y_{4})=(0,0,1,1), \\
-1 & \text{if }(y_{1},y_{2})=(0,1), \\
\exp \left( \gamma _{2}+z_{51}\right)  & \text{if }%
(y_{1},\ldots,y_{4})=(1,0,0,0), \\
\exp \left( -\gamma _{1}+\gamma _{2}+z_{51}\right)  & \text{if }%
(y_{1},\ldots,y_{4})=(1,0,0,1), \\
\exp \left( z_{51}\right) \left( \exp \left( z_{23}\right) -\exp \left(
z_{53}\right) +1\right)  & \text{if }y=(1,0,1,0,0), \\
\exp \left( z_{21}\right) +\exp \left( z_{41}\right) +\exp \left(
z_{21}+z_{43}\right)  &  \\
\quad \; \;-\exp \left( z_{21}+z_{45}\right) -\exp \left( z_{41}+z_{53}\right)
& \text{if }y=(1,0,1,0,1), \\
\exp \left( z_{21}\right)  & \text{if }(y_{1},\ldots,y_{4})=(1,0,1,1),
\\
0 & \text{otherwise},%
\end{array}%
\right.
\end{equation*}
where the additional superscripts indicate $p=3$ and $T=5$.
The remaining moment functions $m_{y^{(0)}}^{(b,3,5)}$, $m_{y^{(0)}}^{(c,3,5)}$,
$m_{y^{(0)}}^{(d,3,5)}$, $m_{y^{(0)}}^{(e,3,5)}$, $m_{y^{(0)}}^{(f,3,5)}$, $m_{y^{(0)}}^{(g,3,5)}$, and $m_{y^{(0)}}^{(h,3,5)}$
 are displayed in Appendix~B.3.4 of \cite{Honore2022moment}.  

\begin{lemma}
\label{lemma:moments_p3T5} If the outcomes $Y=(Y_{1},Y_{2},Y_{3},Y_{4},Y_5)$
are generated from model \eqref{modelARp} with $p=3$, $T=5$ and true
parameters $\beta _{0}$ and $\gamma _{0}$, then we have for all $y^{(0)} \in
\{0,1\}^3$, $x\in \mathbb{R}^{K\times 5}$, $\alpha \in \mathbb{R}$, and $\xi
\in \{a,b,c,d,e,f,g,h\}$ that
\begin{equation*}
\mathbb{E}\left[ m_{y^{(0)}}^{(\xi ,3,5)}(Y,X,\beta _{0},\gamma _{0})\,\Big|%
\,(Y_{-2},Y_{-1},Y_{0})=y^{(0)},\,X=x,\,A=\alpha \right] =0.
\end{equation*}
\end{lemma}

The proof of this lemma is analogous to that of Lemma~\ref{lemma:moments_p2T4}.
In fact, the structure and derivation of $m_{y^{(0)}}^{(a,3,5)}(y,x,\beta ,\gamma )$ are very similar to those
of $m_{y^{(0)}}^{(a,2,4)}(y,x,\beta ,\gamma )$.

\paragraph{AR(2) models with $T=5$:}
The AR(2) model is a special case of the AR(3) model, that is, the moment conditions in Lemma~\ref{lemma:moments_p3T5}
are also applicable to AR(2) models with $T=5$, we just need to set $\gamma_3=0$. In addition, we can
 construct valid moment functions for the  AR(2) model with $T=5$ by using  (time-shifted versions of) the moment functions
 in Section~\ref{sec:moments_p2T4}. Using the results presented so far then gives a total of twenty valid moment functions
 for AR(2) models with  $T=5$. However, there are four linear dependencies between those, so
the total number of linearly independent moment conditions available for $p=2$ and $T=5$
is  equal to $\ell=16$, in agreement with equation~\eqref{MomentCountGeneral}. See Appendix~B.3.6 of \cite{Honore2022moment}
for details.

\subsection{Computational remarks} 
\label{sec:Computation}

In Section~\ref{Sec: Strategy for exploring and using such moment conditions}
we discussed a strategy for numerically exploring the existence of moment conditions.
For this, we need to verify whether solutions to \eqref{MomentsConditional3}
exist for some fixed numerical values for $y^{(0)}$, $x$, $\theta $,
and $\alpha _{1}$,\ldots ,$\alpha _{Q}$, for some $Q>|\mathcal{Y}|$.
For those fixed values, let ${\bf F}$ be the $Q \times |\mathcal{Y}|$
matrix with entries
$f\big(y\,\big|\,y^{(0)},\,x,\,\alpha
_{q};\,\theta \big)$, 
where rows are labeled by $q \in \{1,\ldots,Q\}$ and columns are labeled by
$y \in \mathcal{Y}$. Finding a non-trivial solution to \eqref{MomentsConditional3} is then
equivalent to finding a $|\mathcal{Y}|$-vector ${\bf m} \neq 0$ such that
$$
    {\bf F} \, {\bf m} = 0 ,
$$
that is, finding an element of the nullspace of ${\bf F}$.
Another equivalent rewriting of this condition is
$$
   {\bf F}' \, {\bf F} \, {\bf m} = 0 .
$$
A solution ${\bf m} \neq 0$ exists if and only if the symmetric semi-definite 
 $|\mathcal{Y}| \times |\mathcal{Y}|$
 matrix 
${\bf F}' \, {\bf F}$ has a zero eigenvalue. 
Furthermore, the number of linearly independent valid moment functions ${\bf m}$
is equal to the number of zero eigenvalues of ${\bf F}' \, {\bf F}$.

Calculating eigenvalues numerically is of course straightforward on
modern computers. However, one needs to be careful here to distinguish true zero eigenvalues
from merely very small eigenvalues. For the numerical counting of 
linearly independent  moment conditions in this paper we have used Mathematica to
calculate the eigenvalues of ${\bf F}' \, {\bf F}$ with a working precision of 1,000 digits,
and we only counted an eigenvalue as equal to zero when it was smaller than $10^{-100}$.
See the accompanying Mathematica notebook ``BinaryChoiceMoments.nb'' for further details.

In Section \ref{sec:FindSolutions} we explain how to find analytical expressions
for the moment functions of the panel logit AR(1) model with $T=3$.
For this, one  needs to construct the $8 \times 8$ matrices $\mathbf{B}^{(0)}$ and $ \mathbf{B}^{(1)}$
and the unit vector $\mathbf{e}_4$, and then calculate the moment vectors
$\mathbf{m}^{(0)}$ and $\mathbf{m}^{(1)}$ according to \eqref{FindMomentsAnalytical}.
This calculation and the subsequent simplification of the results for 
$\mathbf{m}^{(0)}$ and $\mathbf{m}^{(1)}$ only takes a few seconds 
on a modern symbolic computation system like
Mathematica.
However, the analogous calculation for more complicated models and  larger
number of time periods $T$ can be slower and may require some additional
human input to speed up the computation.

For example, 
for the derivation of the panel AR(2) moment functions for $T=4$ 
in Appendix~\ref{sec:moments_p2T4} above,
we again followed
the approach described in Section \ref{sec:FindSolutions}, but
to keep the symbolic computation manageable, it was important to first gain some intuition
about the structure of the moment functions by numerical experimentation (i.e.\ using concrete numerical values
to calculate numerical solutions  for the moment functions when imposing various constraints and normalizations on them).
For example, for $m_{y^{(0)}}^{(a,2,4)} $ above, one can first form the conjecture, based on the numerical calculations, that there should exist
a valid moment function of the form
\begin{align}
m_{y^{(0)}}^{(a,2,4)} & =\left\{
\begin{array}{@{\,}l@{\qquad}l}
\mu_1 &
\text{if }y=(0,0,1,0), \\[-2pt]
\mu_2 & \text{if }y=(0,0,1,1),  \\[-2pt]
-1 & \text{if }(y_{1},y_{2})=(0,1),  \\[-2pt]
\mu_3 & \text{if }%
(y_{1},y_{2},y_{3})=(1,0,0),  \\[-2pt]
\mu_4   & \text{if }y=(1,0,1,0), \\[-2pt]
\mu_5 & \text{if }y=(1,0,1,1),  \\[-2pt]
0 & \text{otherwise},%
\end{array}%
\right.
   \label{ConjectureAR2moments}
\end{align}
and one can then solve for the unknown $\mu_1,\ldots,\mu_5$ analytcially by solving the linear system
\begin{equation}
\sum_{y\in \{0,1\}^{4}}\,p_{y^{(0)}}(y,x,\beta _{0},\gamma _{0},\alpha_q
)\;m_{y^{(0)}}^{(a ,2,4)}(y,x,\beta _{0},\gamma _{0})=0 ,
\quad
q \in \{1,2,3,4,5\} ,
    \label{SolveMa24}
\end{equation}
where we can choose five mutually different values $\alpha_q \in \mathbb{R}$ arbitrarily (the solution for $m_{y^{(0)}}^{(a ,2,4)}$ will not
depend on that choice).
Once the structure of the moment function in \eqref{ConjectureAR2moments} 
is found (i.e. guessed correctly via numerical experimentation), then solving \eqref{SolveMa24} to find analytic solutions for \eqref{SolveMa24}
is again very quick.

\subsection{More results for the model with heterogeneous time trends}\label{Additional Calculations for Time Trends}}

In Section \ref{Panel logit AR(1) with heterogeneous time trends} we
have already introduced and discussed the panel logit AR(1) with heterogeneous time trends.
Analogous to Section~\ref{SEC: On the number of moment conditions},
we now  derive a lower bound on the number of valid moment conditions
in this model with heterogeneous time trends.
The probability distribution for $Y_{i}=(Y_{i1},\ldots ,Y_{iT})
$ (conditional on $Y_{i0}$, $X_{i}$, $D_{i}$, $A_{i}$) is given by
\begin{equation*}
f\big(y\,\big|\,y^{(0)},\,x,\,\alpha,\delta ;\,\beta,\gamma\big)\ =\prod_{t=1}^{T}\frac{%
\exp \left( x_{t}^{\prime }\,\beta +y_{t-1}\,\gamma +t\,\delta +\alpha \right)
^{y_{t}}}{1+\exp \left( x_{t}^{\prime }\,\beta +y_{t-1}\,\gamma +t\,\delta
+\alpha \right) },
\end{equation*}
Defining  $a=\exp (\alpha )$, $\pi _{t}(y_{t-1})=\exp [x_{it}^{\prime
}\,\beta +y_{i,t-1}\,\gamma ]$ and  $d=\exp (\delta )$ and noticing that $%
\exp [\alpha +\delta \,(t-1)]=a\,d^{\,t-1}$, we have
\begin{align*}
f\big(y\,\big|\,y^{(0)},\,x,\,\alpha,\delta ;\,\beta,\gamma \big)& =\left[ a\,\pi
_{1}(y_{0})\right] ^{y_{1}}\prod_{t=2}^{T}\left\{ [1+a\,d^{\,t-1}\,\pi
_{t}(1-y_{t-1})]\,\left[ a\,d^{\,t-1}\,\pi _{t}(y_{t-1})\right]
^{y_{t}}\right\}  \\
& =:\sum_{k=0}^{2T-1}\sum_{\ell =\ell _{\min }(k)}^{\ell _{\max
}(k)}\,a^{k}\,d^{\,\ell }\,c_{k\ell }(y)
\end{align*}%
where $c_{k\ell }(y)$ is implicitly defined by the last equation, and
\begin{align*}
\ell _{\min }(k)& :=\lfloor k/2\rfloor \big(\lfloor k/2\rfloor +\mathrm{mod}%
(k,2)\big), \\
\ell _{\max }(k)& :=T\,k-\lfloor (k+1)/2\rfloor \big(\lfloor (k+1)/2\rfloor +%
\mathrm{mod}(k+1,2)\big) 
\end{align*}%
are the minimum and maximum power $d^{\,\ell }$ that are possible in the
above polynomial for a given power $a^{k}$.

Finding a function $m(y)$ that satisfies
$\sum_{y\in \mathcal{Y}}m(y)\,f\big(y\,\big|\,y^{(0)},\,x,\,\alpha,\delta;\,\beta,\gamma \big)=0$
for all $\alpha \in \mathbb{R}$ and $\delta \in \mathbb{R}$ (or equivalently
for all $a>0$ and $d>0$) is therefore equivalent
to finding a function $m(y)$ that satisfies 
$\sum_{y\in \mathcal{Y}}m(y)\, c_{k\ell }(y) = 0$ for all the values
that the indices $k$ and $\ell$ can take in the above expressions
for $f\big(y\,\big|\,y^{(0)},\,x,\,\alpha,\delta ;\,\beta,\gamma\big)$.
The total number of different
values that $k$ and $\ell$ can take is given by
\begin{align*}
& \sum_{k=0}^{2T-1}\left[ 1+\ell _{\max }(k)-\ell _{\min }(k)\right]  \\
& =2T+T\,\sum_{k=0}^{2T-1}k
\\ & \qquad
-\sum_{k=0}^{2T-1}\underbrace{\left\{ \lfloor
k/2\rfloor \big(\lfloor (k+1)/2\rfloor \big(\lfloor (k+1)/2\rfloor +\mathrm{%
mod}(k+1,2)\big)+\lfloor k/2\rfloor +\mathrm{mod}(k,2)\big)\right\} }%
_{=k(k+1)/2} \\
& =2T+T^{2}\,(2T-1)+\frac{1}{3}\,T\,(2T+1)\,(2T-1) \\
& =\frac{T}{3}\left( 2T^{2}-3T+7\right) =:   r_T
\end{align*}%
Thus, the condition $\sum_{y\in \mathcal{Y}}m(y)\, c_{k\ell }(y) = 0$ 
is imposing $r_T$ linear restrictions on the $2^T$ variables $m(y) \in \mathbb{R}$.
If $2^{T}>r_T $, then a solution needs to exist.
For $T\geq 9$ one finds that indeed $2^{T}>r_T $, which implies
that for $T\geq 9$ it must be the case that moment conditions for this model exist.

We have not calculated any of those moment conditions for $T \geq 9$.
However,
as already mentioned in the main text,
for the special case $\beta=0$ (or equivalently $x_t=0$ for all $t$),
where the only remaining common model parameter is $\gamma \in \mathbb{R}$,
we have calculated valid moment functions for this model when $T=8$. There are two such 
valid moment functions $m(y,\gamma)$ and we have obtained analytical expressions for both 
of them. 
The resulting expressions are reported in the accompanying Mathematica program
``BinaryChoiceMoments.nb''. 
All the entries of both of these moment functions are polynomials in $\exp(\gamma)$,
but the polynomial order can be as high as 79, 
making it impossible to report them here. The fact that  expressions for these moment functions
could still be obtained analytically nevertheless illustrates the wide applicability of the methods described
in this paper.

\end{document}